\documentclass[preprint]{elsarticle}

\usepackage{lineno,hyperref}
\modulolinenumbers[5]

\journal{Journal of signal processing}

\usepackage{amsmath}
\usepackage{hyperref}
\usepackage{amssymb,amsthm,amsfonts,mathrsfs}
\usepackage{graphicx}
\usepackage{dsfont,enumitem,stackrel}
\usepackage{color}
\usepackage{mathtools,amssymb,bm}
\usepackage{amstext}
\usepackage{array}
\usepackage[ruled]{algorithm}
\usepackage{algorithmic}
\usepackage{cases}
\usepackage{pgfplots}
\usepackage{caption}
\usepackage{float}
\usepackage{standalone}
\usepackage{subcaption}
\usepackage{tikz}
\usetikzlibrary{shapes,arrows}

\theoremstyle{plain}
\newtheorem{thm}{\textbf{Theorem}}
\newtheorem{lem}{\textbf{Lemma}}
\newtheorem{prop}{\textbf{Proposition}}

\theoremstyle{definition}
\newtheorem{defn}{\textbf{Definition}}

\theoremstyle{remark}
\newtheorem*{rem}{\bf Remark}

\newcommand*{\rom}[1]{\expandafter\@slowromancap\romannumeral #1@}
\graphicspath{{./Figures/}}
\begin{document}

\begin{frontmatter}

\title{ Demixing Sines and Spikes Using Multiple Measurement Vectors}

\author[1]{Hoomaan Maskan \fnref{fn1}}
\ead{hoomaan.maskan@umu.se}
\author[2]{ Sajad Daei\fnref{fn2}}
\ead{Sajad.Daei@eurecom.fr}
\author[3]{Mohammad Hossein Kahaei\corref{cor1}\fnref{fn3}}
\ead{kahaei@iust.ac.ir}
\cortext[cor1]{Corresponding author}
\fntext[fn1]{Department of Mathemetics $\&$ Mathematical Statistics, Umeå University}
\fntext[fn2]{EURECOM- Communication Systems}
\fntext[fn3]{School of Electrical Engineering, Iran University of Science $\&$ Technology}

\begin{abstract}
We address the line spectral estimation problem with multiple measurement corrupted vectors. Such scenarios appear in many practical applications such as radar, optics, and seismic imaging in which the measurements can be modeled as the sum of a spectrally sparse and a block-sparse signal known as outlier. Our aim is to demix the two components and for this purpose, we design a convex problem whose objective function promotes both of the structures. Using the Positive Trigonometric Polynomials (PTP) theory, we reformulate the dual problem as a Semidefinite Program (SDP). Our theoretical results state that for a fixed number of measurements $N$ and constant number of outliers, up to $\mathcal{O}(N)$ spectral lines can be recovered using our SDP problem as long as a minimum frequency separation condition is satisfied. Our simulation results also show that increasing the number of samples per measurement vectors reduces the minimum required frequency separation for successful recovery.
\end{abstract}


\end{frontmatter}

\nolinenumbers

\section{Introduction}\label{section1}
Spectral super resolution is the problem of estimating the spectrum of a signal composed of sinusoids using finite number of samples. This problem, also known as line spectral estimation, is of great importance in signal processing applications such as radar \cite{Heckel2018,sayyari2020blind,jirhandeh2021super,daei2022blind}, multi-path channel estimation \cite{hezaveh2020ofdm}, seismic imaging \cite{Borcea_2002}, and magnetic resonance imaging \cite{Koochakzadeh2018}.

There exist three main attitudes toward spectral super resolution problems: non-parametric methods, parametric approaches \cite{stoica2005spectral}, and optimization-based methods \cite{bhaskar2013atomic,razavikia2019reconstruction,valiulahi2019two}. Periodogram as a non-parametric method can localize sinusoids up to a limited resolution \cite{harris1978use} in the noiseless case. Multiple Signal Classification (MUSIC) is a parametric method which can recover sinusoids perfectly  \cite{schmidt1986multiple}. However, the performance of this method degrades in the presence of noise or outliers. Also, MUSIC needs the correlation matrix of the signal and lack of measurements can highly affect the performance of MUSIC. Other examples of parametric approaches are Estimation of Signal Parameters via Rotational Invariance Technique (ESPRIT) \cite{roy1989esprit} and Matrix Pencil method  \cite{sarkar1995using}. Optimization based approaches minimize the continuous counterpart of the ${\ell_1}$ norm known as the Total Variation (TV) norm \cite{bhaskar2013atomic}. These methods are shown to be robust against Gaussian noise\cite{bhaskar2013atomic}. However, their performance degrades when outliers are present. Tang et.al. proposed a mathematical formulation for the spectral super resolution problem using Atomic Norm Minimization (ANM) \cite{tang2013compressed}. For more illustration, consider a time dispersive multipath channel. The problem is to estimate channel delays and the corresponding complex coefficients using a limited number of pilots. This problem is studied using spectral super resolution and ANM \cite{hezaveh2020ofdm,9066889}.

In most applications, an array of sensors is utilized to receive the signal. In real scenarios, the output of some sensors might be corrupted by perturbations and this makes it harder to super resolve the spectrum of the signal. Thus, the received signal can be described as a mixture of the transmitted signal and spiky noise. This noise can be due to the interference arising from other signals, lightning discharges, and sensor failures. The problem of estimating the transmitted signal from the latter mixture is known as demixing sines and spikes. The demixing problem using the single measurement vector (SMV) is studied in \cite{fernandez2017demixing} and \cite{bayat2020separating}. In some certain settings in applications, we are allowed to collect Multiple Measurement Vectors (MMVs). For example, in Direction Of Arrival (DOA) estimation in array processing \cite{park2018multiple}, the aim is to estimate the DOAs of narrowband sources by observing the output of a sensor array (a group of sensors) during a time window. As each sensor collects a measurement vector (takes snapshot) at each time instance, we have access to MMVs in a time interval. As mentioned earlier, sensors might be exposed to perturbations which can lead to corrupted measurements (interpreted as outliers). To jointly estimate the sources and the ourliers, one could use multiple disjoint SMV demixing problems (corresponding to multiple snapshots) or a single large SMV problem by increasing the array size. However, these approaches do not seem to be reasonable due to cost limitations and array physical constraints. Therefore, it is necessary to exploit the temporal redundancy contained in the MMVs by assuming that the sources remain fixed in a time interval.

In this work, the benefits of using MMVs in the demixing problem are investigated. It is shown that using MMVs makes it possible to localize the sines with high precision. According to the fact that the measurement vectors share the same spectral characteristic of the signal of interest, it is possible to use this joint spectral sparsity and distinguish the signal of interest from the outliers. According to the applied signal model, a new method for spectral super resolution in the presence of outliers is proposed. Also, due to the infinite dimensionality of the TV norm minimization problem, the dual problem is investigated. Using positive trigonometric polynomials (PTP) theory \cite{dumitrescu2007positive}, a tractable SDP is proposed. A vector dual polynomial is formed using the dual variables of the latter SDP. Also, a sufficient condition for the exact recovery of the proposed method is provided. \par

The rest of the paper is as follows: In Section \ref{section2} the demixing problem for the MMV case is formulated, in Section \ref{section3} the TV norm minimization is applied to distinguish the signal of interest from the outliers, in Section \ref{section4} the dual problem is investigated and a new SDP is proposed, in Section \ref{section_noise} dense Gaussian perturbation is added to the model and the corresponding SDP is proposed. Section \ref{section6} presents the numerical results, Section \ref{section7} provides the proof for the main theorem, and Section \ref{section8} is devoted to the conclusion and future work discussions. Also, the proof of the main theorem can be found in Section \ref{section9}. \par

\textbf{Notation}. Throughout this paper, scalars are denoted by lowercase letters, vectors by lowercase boldface letters, and matrices by uppercase boldface letters. The ${i}$th element of the vector ${\boldsymbol x}$ is given by ${\boldsymbol x_i}$. ${|.|}$ denotes cardinality for sets and absolute value for scalars. ${f^{(i)}(t)}$ denotes the ${i}$th derivative of ${f(t)}$ with respect to ${t}$. Transpose, conjugate, and hermitian of a matrix or vector are given by ${(.)^T}$, ${(.)^{\ast}}$, and ${(.)^H}$ respectively.
\section{Problem Formulation}\label{section2}

Suppose that the signal of interest is composed of ${K}$ complex exponentials
\begin{align}\label{eqn1}
s_{jl} & =\sum_{k=1}^{K}a_{kl}e^{i2\pi jf_k}\:\:\:,j\in \mathcal{N},l\in \mathcal{L},
\end{align}
where ${\mathcal{N}=\{0,\ldots,N-1\}}$, ${\mathcal{L}=\{1,\ldots,L\}}$, ${a_{kl}\in \mathbb{C}}$ is the complex amplitude corresponding to the ${k}$th frequency, ${i=\sqrt{-1}}$, ${N}$ is the length of the sinusoids, ${L}$ is the number of measurements or snapshots taken over time, and ${f_k\in\mathbb{T}}$ where ${\mathbb{T}:=\{f_1,\ldots,f_K\}\subset[0,1]}$ is the support set of the signal. In the Fourier domain, (\ref{eqn1}) can be expressed as
\begin{align}\label{eqn2}
G_l(f) & =\sum_{k=1}^{K}a_{kl}\delta(f-f_k),
\end{align}
where ${\delta(f-f_k)}$ is Dirac delta function located at ${f_k}$. The signal can be expressed in a matrix form ${\boldsymbol S}$ whose columns denote the measurements for one snapshot and the rows correspond to the output of each sensor for different snapshots. Note that we can write

\begin{align}
s_{jl} &= \sum_{k=1}^{K}a_{kl}e^{i2\pi jf_k}=\int_{0}^{1}e^{i2\pi jf}G_l(df)=(\mathcal{F}_NG_l)_j, \nonumber
\end{align}
where ${\mathcal{F}_N}$ maps the measure ${G_l}$ to its first N Fourier series coefficients. Here, we study the full measurement case. The results can be extended to the random sampling case \cite{yang2016exact}. \par

As stated in Section \ref{section1}, outliers degrade the performance of recent optimization-based spectral super resolution methods. In order to overcome this problem, the effect of the outliers should be considered in the initial model used for the received signal. Following the same approach of \cite{fernandez2017demixing}, the outliers are added to the received signal as a matrix ${\boldsymbol Z}$
\begin{align}\label{eqn3}
\boldsymbol Y & = \boldsymbol S + \boldsymbol Z = [(\mathcal{F}_NG_1),\ldots,(\mathcal{F}_NG_L)] +  \boldsymbol Z,
\end{align}
where ${\boldsymbol Y_{jl}}$ and ${\boldsymbol Z_{jl}}$ are the received signal and the outliers at ${j}$th sensor and ${l}$th snapshot respectively.{ Note that the outliers affect few number of sensors such that the outlier matrix $\bm{Z}_{N\times L}$ is considered to be row-sparse and ${\boldsymbol \Omega\subset\{0,\ldots,N-1\}}$ denotes the overall support set of the outliers showing the rows of $\bm{Z}_{N\times L}$ with nonzero $\ell_2$ norms.}

\section{Total Variation Norm Minimization}\label{section3}
Without any prior assumption, the demixing problem is ill-posed. Sparse assumption on the signal structure is proved to be helpful in solving linear inverse problems. In compressed sensing theory, \textit{Restricted-Isometry Property (RIP)} guaranteed that a random sampling operator would preserve most of the signal's energy with high probability (see e.g. \cite{daei2018optimal,daei2018sample,daei2019error,daei2019living} for more details). However, in spectral super resolution, it is possible that the non-zero spectral information of the signal lies in the null space of the sampling operator. Thus, an additional condition called the \textit{minimum separation condition} should be met\cite{fernandez2016super}.

\begin{defn}\label{definition1}
		(Minimum separation) Consider the set ${\mathbb{T}}$ as the set of support. The minimum separation is defined as the minimum  wrap-around distance between any elements of ${\mathbb{T}}$,$${\Delta:=\Delta(\mathbb{T})=\underset{(f_1,f_2)\in \mathbb{T}:f_1\neq f_2}{\inf}|f_2-f_1|}.$$
	{For clarification, the wrap-around distance between $f_1=\frac{1}{5}$ and $f_2=\frac{4}{5}$ is equal to $\frac{2}{5}$.}
\end{defn}
In compressed sensing theory, the ${\ell_{1,2}}$ norm was used to promote group sparsity of the received signals sharing the same support sets (see $e.g.$ \cite{daei2019distribution,daei2019exploiting}). The continuous counterpart of ${\ell_{1,2}}$ norm is the group Total Variation (gTV) norm
\begin{align}
\|\boldsymbol X\|_{\rm gTV}&:=\underset{\underset{\boldsymbol F:\mathbb{T}\rightarrow\mathbb{C}^L}{\|\boldsymbol F(t)\|_2\leq1,t\in\mathbb{T}}}{\sup}\sum_{l=1}^{L}Re\{\int_{\mathbb{T}}{\boldsymbol F_l^{H}(t)}\boldsymbol X_l(dt)\}.\nonumber
\end{align}

Fernandez proved that a minimum separation of ${\frac{2.52}{N-1}}$ has to be met so that the gTV norm minimization achieves exact recovery \cite{fernandez2016super}. Following the same insight of \cite{fernandez2017demixing}, we propose the following optimization problem for demixing in the MMV case
\begin{align}\label{eqn4}
\underset{\tilde{\boldsymbol G},\tilde{\boldsymbol Z}}{\min} & \|\tilde{\boldsymbol G}\|_{gTV}+\lambda\|\tilde{\boldsymbol Z}\|_{1,2}\:\:\: s.t.\:\: \boldsymbol Y = [\mathcal{F}_N\tilde{ G_1},\ldots,\mathcal{F}_N\tilde{ G_L}] + \tilde{\boldsymbol Z},
\end{align}
where ${\lambda>0}$ is a regularization parameter, ${\|.\|_{1,2}}$ denotes the matrix ${l_{1/2}}$ norm  and $\mathcal{F}_N$ is the linear operator mapping a vector to its $N$ lowest Fourier Coefficients. The main contribution of this paper is to show that under certain assumptions, the above problem has a unique solution.

	\begin{thm}\label{theorem1}
		Consider ${N}$ measurements with ${L}$ snapshots and suppose that the $\ell_2$ norm of each row in the outlier matrix $\boldsymbol Z$ is non-zero with probability $\frac{s}{N}$. Let the elements of the support set $\mathbb{T}$ satisfy the minimum separation condition of ${\Delta\ge \Delta_{min}=\frac{2.52}{N-1}}$.  If the phases of  ${\boldsymbol a_{l}}:=[a_{1,l},...,a_{K,l}]^T$ (complex amplitudes corresponding to $l$-th snapshot) and the nonzero entries of ${\boldsymbol Z}$ are i.i.d uniformly distributed in ${[0,2\pi]}$, then (\ref{eqn4}) with ${\lambda=1/\sqrt{N}}$ provides the exact solution with probability at least ${1-\epsilon}$ for any $\epsilon>0$ as long as
		 $$K<C_KN\bigg(\log\frac{N}{\epsilon}\bigg)^{-1}\bigg(1+\frac{1}{L}\log\frac{\sqrt{L}N^3}{\epsilon}\bigg)^{-1},$$
		$$s<C_sN\bigg(\log\frac{N}{\epsilon}\bigg)^{-1}\bigg(1+\frac{1}{L}\log\frac{\sqrt{L}N^3}{\epsilon}\bigg)^{-1},$$
		for some constants ${C_K}$, ${C_s}$,  and $N\geq2\times10^{3}$.
\end{thm}

\begin{rem}\label{remark1}
	Using MMVs leads to an increased probability of successful recovery. To see this, consider demixing the corresponding columns of $\bm S$ and $\bm Z$ in the SMV case \cite{fernandez2017demixing}. For a fixed N, each column of $\bm S$ and $\bm Z$ can be recovered from the corresponding column of $\bm Y$  with a probability of at least  $1-\epsilon$. Thus, in order to recover all the columns of $\bm S$, the probability of successful recovery would be at least $1-L\epsilon$. However, in order to solve the problem with a single optimization, as proposed in Theorem 1, the probability of successful recovery for weaker conditions on $N$ and $K$, is at least $1-\sqrt{L}\epsilon$. This explicitly certifies that the proposed method outperforms $L$ individual SMVs in terms of the success probability. It is also worth mentioning that the simulation results in Section \ref{section6} indicate that the MMV performance can actually be better than even a single SMV. This issue can also be captured by Theorem \ref{theorem1} since by increasing L, the conditions on $K, s$ would be weaker. This in turn shows that for fixed $K$ and $s$, the performance of our method enhances by increasing the number $L$ of snapshots.  It is also worth noting that in the SMV case ($L=1$), the bound in Theorem \ref{theorem1} reduces to the required sample complexity of SMV case obtained in \cite[Theorem 2.2]{fernandez2017demixing}.
\end{rem}
The proof of Theorem \ref{theorem1} appears in Section \ref{section7}. In Section \ref{section4} we look at the dual of (\ref{eqn4}) and reformulate it as an SDP.

\section{Dual Problem}\label{section4}
According to the infinite dimensionality of gTV norm in (\ref{eqn4}), we look at its dual formulation and analyze it. The proposed demixing problem (\ref{eqn4}) is closely related to the atomic norm minimization problem introduced in \cite{bhaskar2013atomic}. Using the fact that our signal of interest is composed of ${K}$ complex exponentials, we can present it sparsely with an atomic set containing ${N}$-dimensional sinusoids. The measurements of each snapshot or time sample form a measurement matrix as in (\ref{eqn3}). As a consequence, it is crucial that we use matrix form atoms to build up our signal. Consider the following atomic set { with $\|.\|_2$ denoting the $l_2$ norm}
$${\mathcal{A}=\{\boldsymbol a(f,\phi) \boldsymbol b^H:f\in[0,1],\phi\in[0,2\pi],\|\boldsymbol b\|_2=1\}}$$
for any ${\boldsymbol b\in \mathbb{C}^{L\times 1}}$ and
$${\boldsymbol a(f,\phi)=\frac{1}{\sqrt{N}}e^{i\phi}[1,e^{i2\pi f},\ldots,e^{i2\pi (N-1)f}]^T\in\mathbb{C}^{N}.}$$
Using the above definition of the atomic set, we can define the matrix ${\boldsymbol S}$ as

\begin{align}\label{eqn5}
\boldsymbol S & =\sqrt{N}\sum_{k=1}^{K}\boldsymbol a(f_k,\phi)\boldsymbol \psi_k^H=\sum_{k=1}^{K}c_k\boldsymbol a(f_k,\phi)\boldsymbol b_k^H,
\end{align}
where ${c_k=\sqrt{N}\|\boldsymbol \psi\|_2>0}$ and  ${\boldsymbol b_k=c_k^{-1}\boldsymbol \psi_k\sqrt{N}}$ with ${\|\boldsymbol b_k\|_2=1}$.
According to \cite{tang2013compressed,chandrasekaran2012convex}, spectral super resolution problem can be treated using atomic norm minimization. This attitude arises from the fact that in spectral super resolution problem the spectrum of the signal of interest is sparse. The atomic norm is defined as
\begin{align}\label{eqn6}
\|\boldsymbol X\|_{\mathcal{A}} & :=\inf\{t>0:\boldsymbol X\in tconv(\mathcal{A})\},
\end{align}
where ${conv(\mathcal{A})}$ denotes the convex hull of the atomic set ${\mathcal{A}}$.

Using the definition of the atomic norm, (\ref{eqn4}) can be represented as
\begin{align}\label{eqn7}
\underset{\tilde{\boldsymbol S},\tilde{\boldsymbol Z}}{\min} & \|\tilde{\boldsymbol S}\|_{\mathcal{A}}+\lambda\|\tilde{\boldsymbol Z}\|_{1,2}\:\:\: s.t.\:\: \boldsymbol Y =\tilde{\boldsymbol S} + \tilde{\boldsymbol Z}.
\end{align}
In order to formulate the dual problem, we need the definition of dual atomic norm as
\begin{align}
\|\boldsymbol \Gamma \|_{\mathcal{A}}^{\ast} & = \underset{\|\tilde{\boldsymbol S}\|_{\mathcal{A}}\leq1}{\sup} <\boldsymbol \Gamma,\tilde{\boldsymbol S}>_{\mathbb{F}} ,\nonumber\\
& = \underset{\underset{\|\boldsymbol b\|_2=1}{\underset{\phi\in[0,2\pi]}{f\in[0,1]}}}{\sup}<\boldsymbol \Gamma,e^{i\phi}\boldsymbol a(f,0)\boldsymbol b^H>_{\mathbb{F}} ,\nonumber\\
& = \underset{\underset{\|\boldsymbol b\|_2=1}{f\in[0,1]}}{\sup} |<\boldsymbol \Gamma,\boldsymbol a(f,0)\boldsymbol b^H>_{\mathbb{F}}|, \nonumber\\
& = \underset{f\in[0,1]}{\sup} \| \boldsymbol \Gamma^H\boldsymbol a(f)\|_2, \nonumber
\end{align}
where ${<.>_{\mathbb{F}}}$ shows the Frobenius inner product. Using the above definition, the dual of (\ref{eqn7}) is written as
\begin{eqnarray}\label{eqn9}
\underset{\boldsymbol \Gamma\in\mathbb{C}^{N\times L}}{\max}Re< {\boldsymbol Y}, \boldsymbol \Gamma >_{\mathbb{F}}\:  s.t. && \underset{f\in[0,1]}{\sup} \|{\boldsymbol \Gamma}^H \boldsymbol a(f,0)\|_2  \leq1,\nonumber \\
&& \|\boldsymbol \Gamma\|_{\infty,2}\leq\lambda,
\end{eqnarray}
where ${Re<.>}$ denotes the real part of the inner product and ${\|.\|_{\infty,2}}$ is the matrix infinity/2 norm defined as
$${\|\boldsymbol \Gamma\|_{\infty,2}=\underset{i}{\max}\|\boldsymbol \Gamma_{i,:}\|_2.}$$
By applying the PTP theory \cite{dumitrescu2007positive}, the maximization constraint in (\ref{eqn9}) can be reformulated as a Linear Matrix Inequality (LMI) given by
\begin{eqnarray}\label{eqn10}
\underset{\boldsymbol \Gamma\in\mathbb{C}^{N\times L},\boldsymbol \Lambda \in\mathbb{C}^{N\times N} }{\max}Re< {\boldsymbol Y}, \boldsymbol \Gamma >_{\mathbb{F}}\:  s.t. && \left[
\begin{array}{cc}
\boldsymbol \Lambda & \boldsymbol \Gamma \\
\boldsymbol \Gamma^{H} & \boldsymbol I_{L} \\
\end{array}
\right]\succeq0, \nonumber\\
&& \mathcal{T}^{\ast}(\boldsymbol \Lambda)=\left[
\begin{array}{c}
1 \\
\boldsymbol 0 \\
\end{array}
\right],
\nonumber\\
&& \|\boldsymbol \Gamma\|_{\infty,2}\leq\lambda,
\end{eqnarray}
where ${\mathcal{T^{\ast}}}$ is defined as
$${\mathcal{T}^{\ast}(\boldsymbol \Lambda)_j=\sum_{i=1}^{N-j+1}\boldsymbol \Lambda_{i,i+j-1},}$$
${\boldsymbol I_L}$ denotes the identity matrix of size ${L\times L}$, ${\boldsymbol 0\in\mathbb{C}^{N-1}}$ is a zero vector, and ${\succeq0}$ denotes positive semi-definiteness. \par

In order to localize the frequencies of the signal of interest and the noisy spikes, Lemma\ref{lemma1} is presented.
\begin{lem}\label{lemma1}
	The solution to (\ref{eqn7}) is unique if for ${\boldsymbol \Gamma\in\mathbb{C}^{N\times L}}$ and the vector-valued dual polynomial ${\boldsymbol Q=\boldsymbol a(f,0)^H\boldsymbol \Gamma}$, we have
	\begin{subequations}
		\begin{eqnarray}
		\boldsymbol Q(f_k)&=& \frac{c_k}{|c_k|}\boldsymbol b_k^{H}\:\:\: {  \text{ for }k \text{ s.t. }f_k\in\mathbb{T},}\label{eqn11}\\
		\|\boldsymbol Q(f_j)\|_2<1& &\:\:\:\forall f_j\in[0,1]\backslash \mathbb{T},\label{eqn12}
		\end{eqnarray}
		\text{and for any }${d}\in{\boldsymbol \Omega}$ \text{ and } ${l \in \boldsymbol \Omega^c}$,
		\begin{eqnarray}
		\boldsymbol \Gamma_{ d,:} &=& \lambda\frac{\boldsymbol Z_{d,:}}{\|\boldsymbol Z_{d,:}\|_2}, \label{eqn13}\\
		\|\boldsymbol \Gamma_{l,:}\|_{2} &<& \lambda.\label{eqn14}
		\end{eqnarray}
	\end{subequations}
\end{lem}
\begin{proof}
	If we find a ${\boldsymbol \Gamma}$ satisfying the above conditions, it is dual feasible. Consider ${\hat{\boldsymbol S}}$ and ${\hat{\boldsymbol Z}}$ as the solutions to (\ref{eqn7}). Then, we would have
	\begin{eqnarray}
	\|\hat{\boldsymbol S}\|_{\mathcal{A}} &\geq& \|\hat{\boldsymbol S}\|_{\mathcal{A}}\|{\boldsymbol \Gamma}\|_{\mathcal{A}}^{\ast}\nonumber \\
	&\geq& <\boldsymbol \Gamma,\hat{\boldsymbol S}>_{\mathbb{R}} \nonumber\\
	&=& <\boldsymbol \Gamma,\sum_{k=1}^{K}c_k\boldsymbol a(f_k,\phi_k)\boldsymbol b_k^H>_{\mathbb{R}} \nonumber\\
	&=& \sum_{k=1}^{K}Re\{c_k^{\ast}<{\boldsymbol \Gamma},\boldsymbol a(f_k,\phi_k)\boldsymbol b_k^H>\} \nonumber\\
	&=& \sum_{k=1}^{K}Re\{c_k^{\ast}<\boldsymbol b_k,\boldsymbol Q(f_k)^H>\}\nonumber\\
	&=& \sum_{k=1}^{K}Re\{c_k^{\ast}\frac{c_k}{|c_k|}\}\geq \|\hat{\boldsymbol S}\|_{\mathcal{A}}.\nonumber
	\end{eqnarray}
	Also,
	\begin{eqnarray}
	Re<\hat{\boldsymbol Y},{\boldsymbol \Gamma}> &=& Re<\hat{\boldsymbol S},{\boldsymbol \Gamma}>+Re<\hat{\boldsymbol Z},{\boldsymbol \Gamma
	}> \nonumber\\
	&=&   \|\hat{\boldsymbol S}\|_{\mathcal{A}}+\sum_{ d\in{\boldsymbol \Omega}}^{ }Re\{\hat{\boldsymbol Z_{d,:}^{\ast}},\boldsymbol \Gamma_{d,:} \} \nonumber\\
	&=&  \|\hat{\boldsymbol S}\|_{\mathcal{A}}+ \lambda\sum_{d\in\boldsymbol \Omega}^{ }Re\{\frac{\boldsymbol Z_{d,:}^{\ast}\boldsymbol Z_{d,:}}{\|\boldsymbol Z_{d,:}\|_2}\}=\lambda\|\boldsymbol Z\|_{1,2},\nonumber
	\end{eqnarray}\par
	\noindent where the last equality is derived using \ref{eqn13}. Therefore, we must have

${<\boldsymbol \Gamma,\hat{\boldsymbol S}>_{\mathbb{R}}= \|\hat{\boldsymbol S}\|_{\mathcal{A}}+\lambda\|\hat{\boldsymbol Z}\|_{1,2}}$. Thus, by strong duality, ${\hat{\boldsymbol S}}$ and ${\hat{\boldsymbol Z}}$ are primal optimal and ${\boldsymbol \Gamma}$ is dual optimal. To investigate uniqueness, we consider $\tilde{\boldsymbol S}=\sum_{k\in\widetilde{{\mathbb{T}}}}^{ }\tilde{c}_k\boldsymbol a(\tilde{f}_k,\tilde{\phi}_k)\tilde{\boldsymbol b}_k^H$ and ${\tilde{\boldsymbol Z}}$ with supports $\widetilde{{\mathbb{T}}}$ and $\widetilde{\Omega}$, respectively as the other optimal solutions to (\ref{eqn7}).
		Then, we get
		\begin{eqnarray}
		& & <\tilde{\boldsymbol Y},\boldsymbol \Gamma>_{\mathbb{R}}=<\tilde{\boldsymbol S},\boldsymbol \Gamma>_{\mathbb{R}}+<\tilde{\boldsymbol Z},\boldsymbol \Gamma>_{\mathbb{R}}\nonumber\\
		&=&\sum_{\tilde{f}_k\in\mathbb{T}\cap \widetilde{\mathbb{T}}}^{ }Re\{\tilde{c}_k<\tilde{b}_k,\boldsymbol Q(\tilde{f}_k)^H>\} \nonumber\\
		&+&\sum_{\tilde{f}_j\in\mathbb{T}^c\cap\widetilde{\mathbb{T}}}^{ }Re\{\tilde{c}_j<\tilde{b}_j,\boldsymbol Q(\tilde{f}_j)^H>\}\nonumber\\
		&+& \sum_{ d \in {\Omega}\cap \widetilde{{\Omega}}}^{ }Re\{<\tilde{\boldsymbol Z}_{d,:}^{\ast},\boldsymbol \Gamma_{d,:}>\}+\sum_{l\in{\Omega}^c\cap \widetilde{\Omega}}^{ }Re\{<\tilde{\boldsymbol Z}_{l,:}^{\ast},\boldsymbol \Gamma_{l,:}>\} \nonumber\\
		&\leq& \sum_{\tilde{f}_k\in\mathbb{T}\cap\widetilde{\mathbb{T}}}^{ }Re\{\tilde{c}_k\|\tilde{\boldsymbol b}_k\|_2\|\boldsymbol Q(\tilde{f}_k)\|_2 + \sum_{\tilde{f}_j\in\mathbb{T}^c\cap\widetilde{\mathbb{T}}}^{ }Re\{\tilde{c}_j\|\tilde{\boldsymbol b}_j\|_2\|\boldsymbol Q(\tilde{f}_j)\|_2\nonumber\\
		&+& \lambda\sum_{ d \in {\Omega}\cap\widetilde{\Omega}}^{ }Re\{\|\tilde{\boldsymbol Z}_{d,:}\|_2\} + \|\boldsymbol \Gamma_{ \Omega^c,:}\|_{\infty,2}\sum_{l\in{\Omega}^c\cap \widetilde{\Omega}}^{ }Re\{\|\tilde{\boldsymbol Z}_{l,:}\|_2\} \nonumber\\
		&<& \sum_{\tilde{f}_k\in\mathbb{T}\cap\widetilde{\mathbb{T}}}^{ }\tilde{c}_k\|\tilde{\boldsymbol b}_k\|_2   +\sum_{\tilde{f}_j\in\mathbb{T}^c\cap\widetilde{\mathbb{T}}}^{ }\tilde{c}_j\|\tilde{\boldsymbol b}_j\|_2   +\lambda\sum_{d \in
			\Omega\cap \widetilde{\Omega}}Re\{\|\tilde{\boldsymbol Z}_{d}\|_2\}\nonumber\\
		&+& \lambda\sum_{\boldsymbol l\in{\Omega}^c\cap \widetilde{\Omega}}^{ }Re\{\|\tilde{\boldsymbol Z}_{l,:}\|_2\} = \|\tilde{\boldsymbol S}\|_{\mathcal{A}}+\lambda\|\tilde{\boldsymbol Z}\|_{1,2},\nonumber
		\end{eqnarray}
		which contradicts the strong duality. Thus, all optimal solutions are solely supported on $\mathbb{T}$ and ${\Omega}$. Since the atoms in ${\mathbb{T}}$ and $\Omega$ are linearly independent, the pair $({\hat{\boldsymbol S}}, {\hat{\boldsymbol Z}})$ is the unique optimal solution to (\ref{eqn7}).
\end{proof}
\section{Demixing in Presence of Dense Perturbation}\label{section_noise}
In many practical scenarios ($e.g.$ Direction of Arrival (DOA) estimation), the presence of dense perturbations is unavoidable \cite{yang2017sparse}. When the received signal is perturbed with dense noise, one can modify (\ref{eqn4}) as a new optimization problem and the corresponding SDP. In this case, the received data is in the form of
\begin{align}
    \boldsymbol Y & = \boldsymbol S + \boldsymbol Z +\boldsymbol W= [(\mathcal{F}_NG_1),\ldots,(\mathcal{F}_NG_L)] +  \boldsymbol Z+\boldsymbol W,\nonumber
\end{align}
where ${\boldsymbol W\in \mathbb{C}^{N\times L}}$ is the additive noise matrix with $i.i.d.$ elements distributed as zero mean Gaussian distribution with standard deviation $\sigma$. Now, by modifying (\ref{eqn4}), we reach

	\begin{align}\label{eqn11noise}
	\underset{\tilde{\boldsymbol G},\tilde{\boldsymbol Z}}{\min} & \|\tilde{\boldsymbol G}\|_{gTV}+\lambda\|\tilde{\boldsymbol Z}\|_{1,2}\:\:\: s.t.\:\: \|\boldsymbol Y - [\mathcal{F}_N\tilde{ G_1},\ldots,\mathcal{F}_N\tilde{ G_L}] - \tilde{\boldsymbol Z}\|_F\leq \eta,
	\end{align}
	where $\eta$ is an upper-bound of $\|\bm{W}\|_F$. In what follows, we derive the dual problem of (\ref{eqn11noise}) and its corresponding semidefinite relaxation.
\begin{lem}
    The dual problem of (\ref{eqn11noise}) is

	\begin{eqnarray}\label{eqn12noisy}
	\underset{\boldsymbol \Gamma\in\mathbb{C}^{N\times L}}{\max}Re< {\boldsymbol Y}, \boldsymbol \Gamma >_{\mathbb{F}}-\eta\|\boldsymbol \Gamma\|_F\:  s.t. && \underset{f\in[0,1]}{\sup} \|{\boldsymbol \Gamma}^H \boldsymbol a(f,0)\|_2  \leq1,\nonumber \\
	&& \|\boldsymbol \Gamma\|_{\infty,2}\leq\lambda,
	\end{eqnarray}
which is equivalent to the following SDP,

	\begin{eqnarray}\label{eqn13noisy}
	\underset{\boldsymbol \Gamma\in\mathbb{C}^{N\times L},\boldsymbol \Lambda \in\mathbb{C}^{N\times N} }{\max}Re< {\boldsymbol Y}, \boldsymbol \Gamma >_{{F}}-\eta\|\boldsymbol \Gamma\|_F\:  s.t. && \left[
	\begin{array}{cc}
	\boldsymbol \Lambda & \boldsymbol \Gamma \\
	\boldsymbol \Gamma^{H} & \boldsymbol I_{L} \\
	\end{array}
	\right]\succeq0, \nonumber\\
	&& \mathcal{T}^{\ast}(\boldsymbol \Lambda)=\left[
	\begin{array}{c}
	1 \\
	\boldsymbol 0 \\
	\end{array}
	\right],
	\nonumber\\
	&& \|\boldsymbol \Gamma\|_{\infty,2}\leq\lambda,
	\end{eqnarray}
with $\boldsymbol 0\in \mathbb{C}^{n-1}$ being a vector of zeros.
\end{lem}
\begin{proof}
    Problem (\ref{eqn11noise}) can be reformulated as
   	\begin{align}
   	\underset{\tilde{\boldsymbol G},\tilde{\boldsymbol Z}}{\min}  \|\tilde{\boldsymbol G}\|_{gTV}+\lambda\|\tilde{\boldsymbol Z}\|_{1,2}\:\:\: s.t.\:\: &\|\boldsymbol Y - \boldsymbol U\|_{\mathbb{F}}^2\leq \eta^2\nonumber \\
   	& \boldsymbol U=[\mathcal{F}_N\tilde{ G_1},\ldots,\mathcal{F}_N\tilde{ G_L}] + \tilde{\boldsymbol Z}.\nonumber
   	\end{align}
The Lagrangian of the above problem is

	\begin{align}\label{eqn14noisy}
	\mathcal{L}(\tilde{\boldsymbol G},\tilde{\boldsymbol Z},\boldsymbol \Gamma)&=\|\tilde{\boldsymbol G}\|_{gTV}-\langle [\tilde{ G_1},\ldots,\tilde{ G_L}],\mathcal{F}_N^*\boldsymbol\Gamma\rangle_{{F}}+\lambda\|\tilde{\boldsymbol Z}\|_{1,2}-\langle\tilde{\boldsymbol Z},\boldsymbol\Gamma\rangle_F\nonumber\\
	& +\langle\boldsymbol U,\boldsymbol \Gamma\rangle_F+\nu\left(\|\boldsymbol Y-\boldsymbol U\|_F^2-\eta^2\right).
	\end{align}
Due to the first constraint in (\ref{eqn12noisy}), $\|\tilde{\boldsymbol G}\|_{gTV}-\langle [\tilde{ G_1},\ldots,\tilde{ G_L}],\mathcal{F}_N^*\boldsymbol\Gamma\rangle_{{F}}$ is minimized for $\tilde{\boldsymbol G}=0$ and because of the second constraint in (\ref{eqn12noisy}), $\lambda\|\tilde{\boldsymbol Z}\|_{1,2}-\langle\tilde{\boldsymbol Z},\boldsymbol\Gamma\rangle_F$ is minimized for $\tilde{\boldsymbol Z}=0$. Next,  considering the convexity of (\ref{eqn14noisy}), we evaluate its gradient w.r.t $\boldsymbol U$ and set the result to zero to get,
    \begin{align}
        & \frac{\partial\mathcal{L}}{\partial\boldsymbol U}=\boldsymbol\Gamma - 2\nu(\boldsymbol Y- \boldsymbol U)=0\rightarrow \boldsymbol U=\boldsymbol Y-\frac{\boldsymbol \Gamma}{2\nu}.\nonumber
    \end{align}
    Therefore,

    	\begin{align}
    	\langle \boldsymbol Y,\boldsymbol \Gamma \rangle_F-\frac{1}{2\nu}\|\Gamma\|_F^2+\nu\left(  \|\frac{\Gamma}{2\nu}\|_F^2-\eta^2  \right)= \langle \boldsymbol Y,\boldsymbol \Gamma \rangle_F-\frac{1}{4\nu}\|\Gamma\|_F^2-\nu\eta^2.
    	\end{align}
Since $\nu$ is positive, taking derivative w.r.t $\nu$, setting the result equal to zero, and plugging the result back leads to (\ref{eqn12noisy}) with the constraints used. Thus, (\ref{eqn12noisy}) is the dual problem of (\ref{eqn11noise}) and using the PTP theory \cite{dumitrescu2007positive}, (\ref{eqn13noisy}) is concluded.
    \end{proof}
\section{Numerical Results}\label{section6}
\subsection{Without Perturbation}
In this subsection, numerical experiments are presented to evaluate the performance of the method proposed in Section \ref{section4}. First, we investigate the constraints (\ref{eqn11}) and (\ref{eqn12}) on the dual polynomial and the constraints (\ref{eqn13}) and (\ref{eqn14}) on the dual variable. Using these constraints, one can localize the signal frequencies and the outliers' spikes. Next, the minimum required frequency separation for successful recovery in the MMV case is compared with the one needed in the SMV case. In all simulations of this subsection, the number of sensors or the signal length is set to ${N=50}$. In the first part of the simulations, the signal of interest ${\boldsymbol S\in\mathbb{C}^{N\times L}}$ has ${K=3}$ frequencies and the coefficients ${a_{kl}}$ are always drawn from a standard i.i.d complex Gaussian distribution. The outliers' spikes are considered to be in ${s=3}$ different random positions in each snapshot. For better visualization, it is assumed that outliers happen in each sensor only once. Figure (\ref{fig1}) depicts ${\|\boldsymbol Q(f)\|_2}$ for ${L=5}$ snapshots and ${\mathbb{T}=\{0.1,0.4,0.8\}}$. As it can be observed, the signal frequencies can be estimated by solving ${\|\boldsymbol Q(f)\|_2=1}$ for all ${f\in[0,1]}$. The outliers are localized in each receiving sensor using (\ref{eqn13}). We considered ${s=3}$ noisy spikes occurring randomly in each measurement without replacement. Thus, with ${L=5}$ we expect to detect ${15}$ outliers in the receiver. Figure (\ref{fig2}) depicts the result. As it turns out, Figure (\ref{fig2}) verifies the conclusion of Lemma \ref{lemma1}.
\begin{figure}[!t]
	\centering
	\includegraphics[scale=0.3]{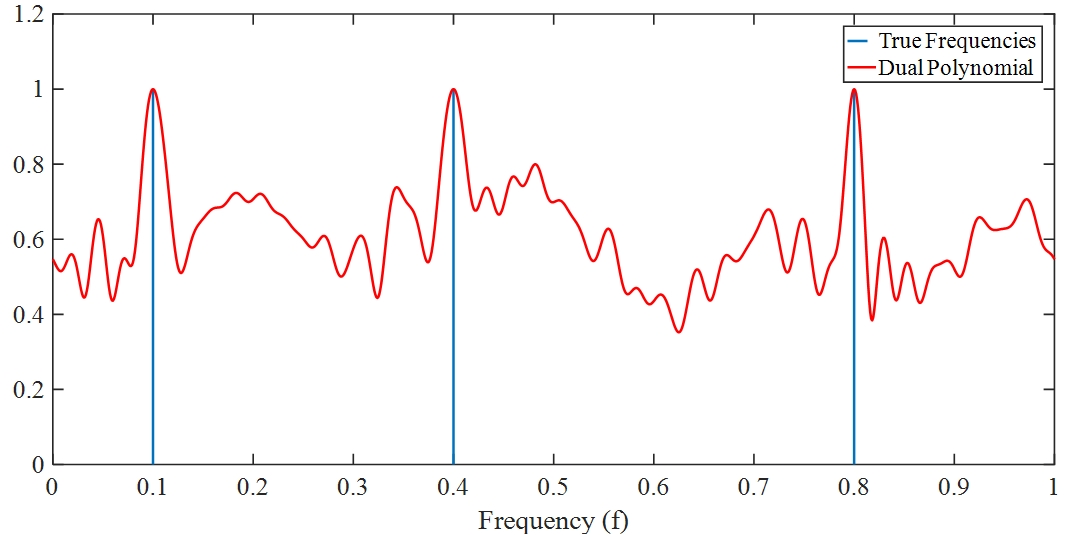}
	\caption  { ${\ell_2}$ norm of dual polynomial and true frequencies.}\label{fig1}
\end{figure}
\begin{figure}[!t]
	\centering
	\includegraphics[scale=0.35]{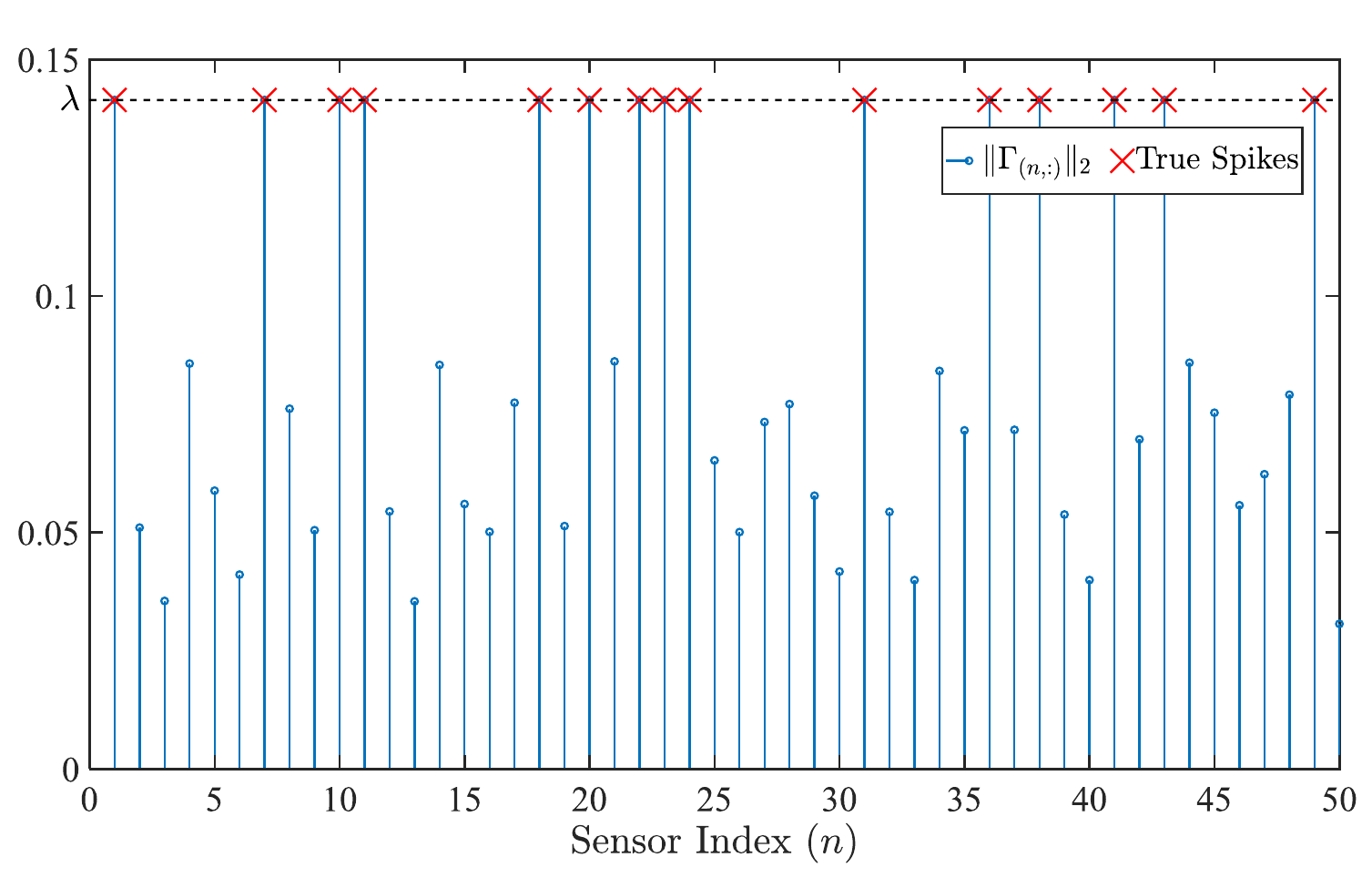}
	\caption{$\ell_2$ norm of  $\bm{\Gamma}$ rows in Lemma 1 in terms of sensor indices. The estimated spike locations are found by identifying indices where the $\ell_2$ norm of $\Gamma$ rows  achieves $\lambda$.
}\label{fig2}
\end{figure}

Next, we investigate the minimum separation condition. To do this, we consider two frequencies slowly taking distance. The first frequency is fixed at ${f_1=0.2}$ and the second one has a distance of ${f_{\delta}=\{0.1/N:0.1/N:1.5/N\}}$ from ${f_1}$. During this experiment, ${s=10}$ outliers out of ${N=50}$ are considered in the overall measurement process. Also, we define ${\boldsymbol f_{est}=[f_1^{est},f_2^{est}]}$ as the estimated frequencies vector. A successful estimation is defined as when
\begin{align}
& \max\{|\boldsymbol f_{est}-\boldsymbol f_{true}|\}\leq 10^{-4}
\end{align}
where ${\boldsymbol f_{true}}$ denotes the true frequencies. With this definition, Figure (\ref{fig3}) illustrates the probability of successful recovery for ${L=\{1,3,5\}}$ over ${100}$ Monte-Carlo simulations.
\begin{figure}[!t]
	\centering
	\includegraphics[scale=0.3]{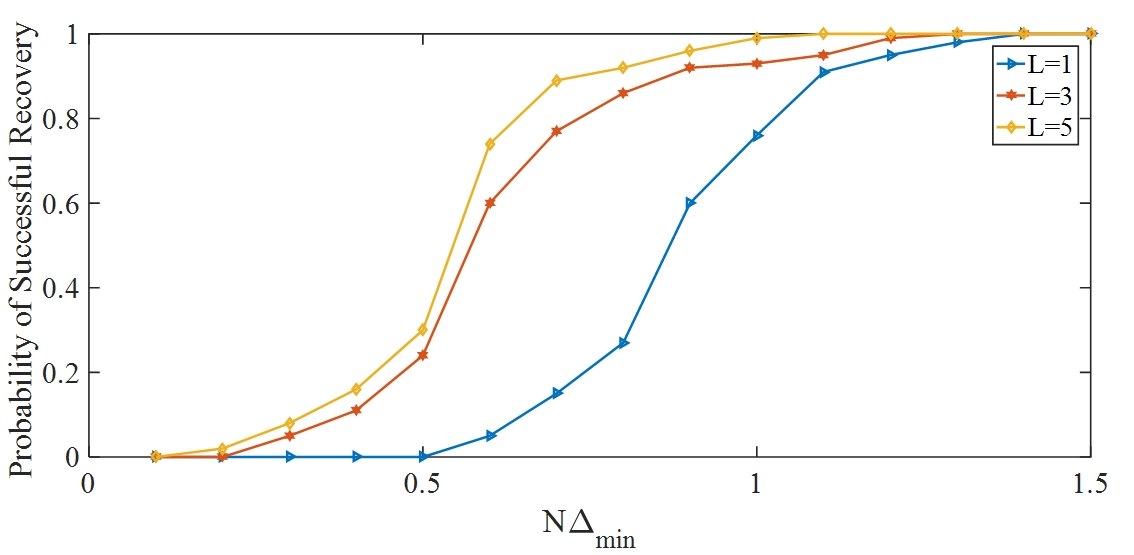}
	\caption{Probability of successful recovery in terms of minimum separation for various number of snapshots.}\label{fig3}
\end{figure}
As seen, the minimum required frequency separation is decreased with an increase in the number of snapshots.
\begin{figure}[!h]
	\centering
	\includegraphics[scale=0.32]{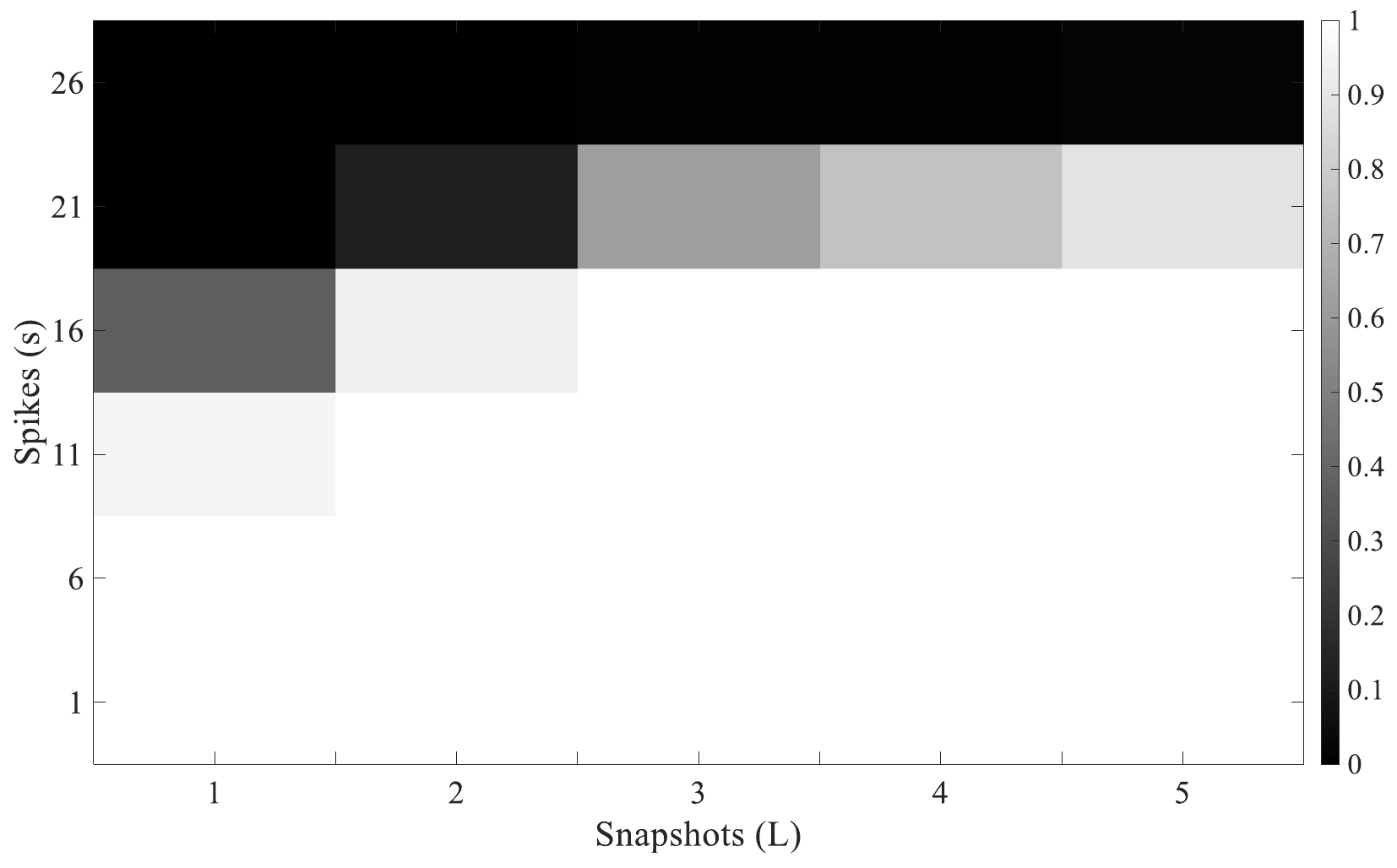}
	\caption{{ Performance of our atomic norm minimization method for diverse values of the number of snapshots (L) and outliers' spikes (s). The heat-map shows the probability of successful recovery (white: success black: failure).}}\label{fig41}
\end{figure}
{To more elaborate on Theorem \ref{theorem1}, the phase transition diagrams for $s$ and $K$ are plotted in Figures \ref{fig41} and \ref{fig4} for  $N=50, s=5$, and $N=50, K=5$, respectively. In each Monte-Carlo simulation, the frequencies are randomly generated satisfying the separation condition. The total number of Monte-Carlo runs is $100$. As Figure \ref{fig41} shows, increasing the number of snapshots can affect the maximum number of possible spikes to recover. However, this behaviour is up to a limited point. This is aligned with our bound shown in Theorem \ref{theorem1}. Moreover, Figure \ref{fig4} shows that for fixed $N$ and $K$, the probability of successful recovery increases when the number of snapshots ($L$) rises. Also, increasing the number of sources $K$ leads to recovery corruption for fixed $L$ and $N$. }

\begin{figure}[!t]
    \centering
    \includegraphics[scale=0.38]{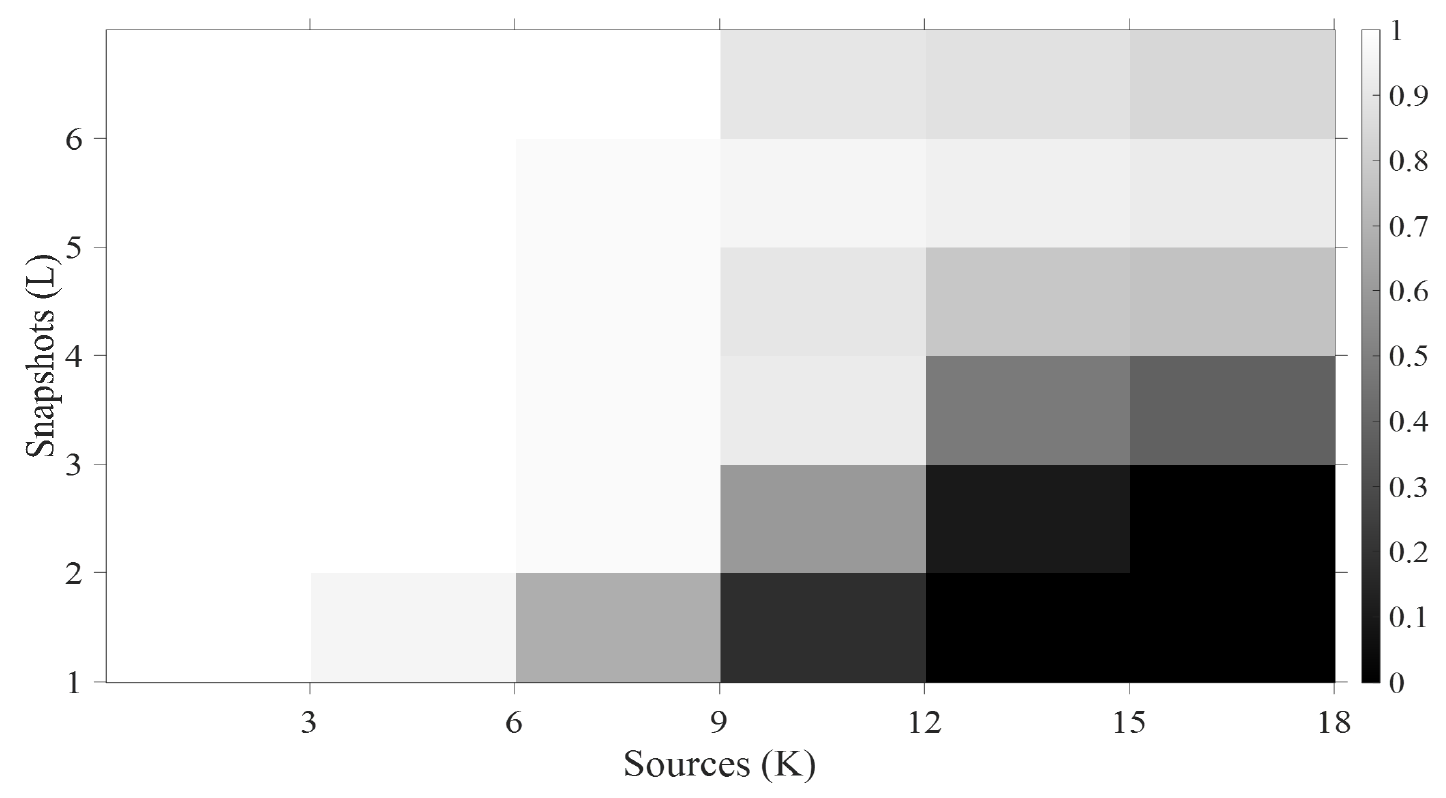}
    \caption{ Performance of our atomic norm minimization method for diverse values of the number of snapshots (L) and sources (K). The heat-map shows the probability of successful recovery (white: success black: failure).}
    \label{fig4}
\end{figure}
\subsection{With Perturbation}
In this subsection, various scenarios are considered and the results are compared with the state of the art SPA \cite{6857424} method. Note that comparison with conventional SPICE method \cite{5617289} is not implemented as this method was designed to estimate on-grid frequencies leading to basis mismatch issues and estimation inaccuracies. Therefore, we only compare our method with SPA \cite{6857424} which is often regarded as a continuous version of SPICE \cite{5617289}. In order to get closer to a more realistic scenario, consider the DOA estimation problem with $K=3$ sources where the first and third sources are considered coherent. Take the incoming directions to be $\boldsymbol f_{true}=[0.1,0.4,0.8]$ and the corresponding estimation vector to be $\boldsymbol f_{est}$. In what follows, we compare our proposed method with the SPA \cite{6857424} when both methods are perturbed with impulsive spiky and Gaussian dense noises.
\subsubsection{Effect of Taking Snapshots}
 To compare the performance of the two methods for different number of snapshots, we set the number of DOA sensors, $N$, to $50$ and the number of spikes, $s$, to $10$. The elements of the noise matrix $\bm{W}$ are considered to be $i.i.d.$ and distributed as Gaussian with zero mean and variance $0.5$. Also, the spectral norm of the spiky noise $\|\boldsymbol Z\|_F$ is set to $25$. This value does not affect the performance of the proposed method but devastates the performance of SPA. The number of snapshots, $L$, is ranged from 5 to 30 and the result is presented in Figure \ref{fig5} for 100 Monte-Carlo simulations. The error is measured in terms of the Mean Square Error (MSE) defined as
\begin{align}
    MSE=\frac{1}{K}\|\boldsymbol f_{est}-\boldsymbol f_{true}\|_2^2.\nonumber
\end{align}
\begin{figure}[!t]
    \centering
    \includegraphics[scale=0.38]{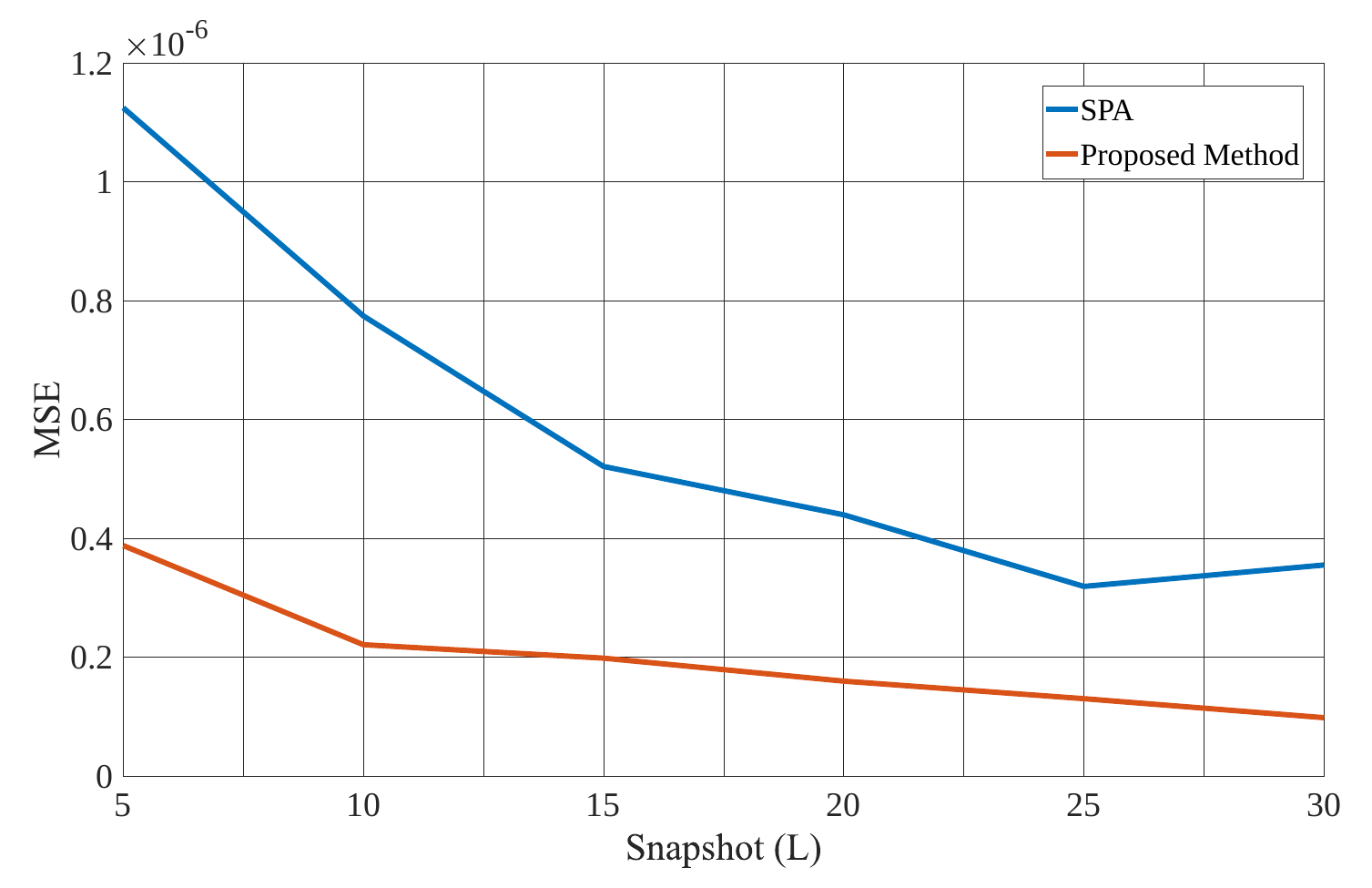}
    \caption{  Comparison of the performance of proposed method over SPA when the number of snapshots increases and the spikes energy is low.}
    \label{fig5}
\end{figure}
As Figure (\ref{fig5}) depicts, both methods show improvements as the number of snapshots increases. However, it is apparent that the proposed method has higher accuracy than the SPA method.
\subsubsection{Effect of Impulsive Noise Energy}
Here, we would like to investigate the performance of two methods for various levels of $\|\boldsymbol Z\|_F$ which is the Frobenius norm of $\boldsymbol Z$. The setting is similar to that of the previous subsection except that the number of snapshots is fixed to $10$. The result is shown in Figure (\ref{fig6}).
\begin{figure}[!t]
    \centering
    \includegraphics[scale=0.49]{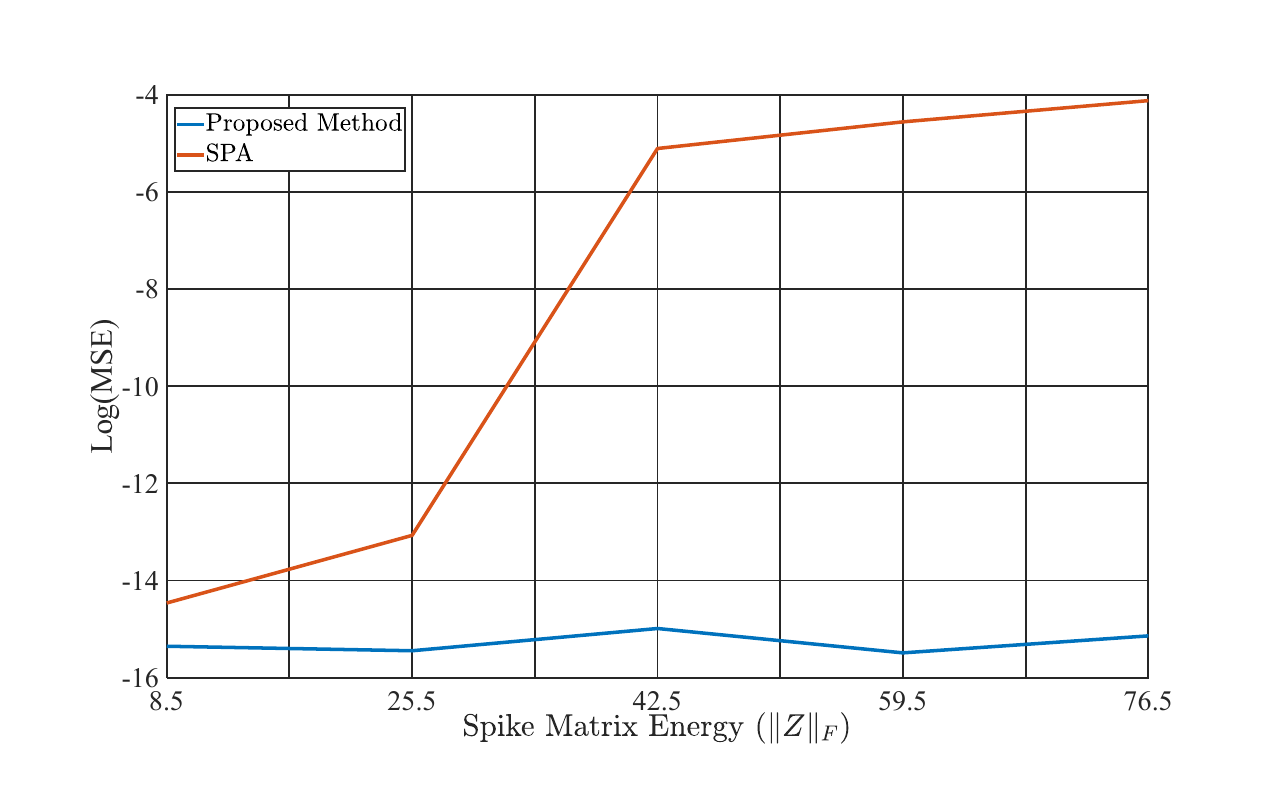}
    \caption{ Comparison of the performance of proposed method over SPA in terms of the energy of spikes $i.e.$ $\|Z\|_F$.}
    \label{fig6}
\end{figure}
As expected, the performance of the proposed method does not change considerably compared to the SPA in this scenario.  However, it can be observed that by increasing the energy of spikes, the performance of SPA is corrupted in such a way that after some threshold a successful recovery will be impossible.
\subsubsection{SNR}
This subsection focuses on the effect of Gaussian noise variance employed in $\bm{W}$ on the performance of both methods in terms of MSE. The simulation is based on the same setting as the first subsection except that the variance of each element in $\bm{W}$ $i.e.$ $\sigma^2$ varies from $0.5$ to $6.5$, $L=5$ and $s=5$ with the spikes energy level of $30$. The result is shown in Figure (\ref{fig7}).
\begin{figure}[!t]
    \centering
    \includegraphics[scale=0.38]{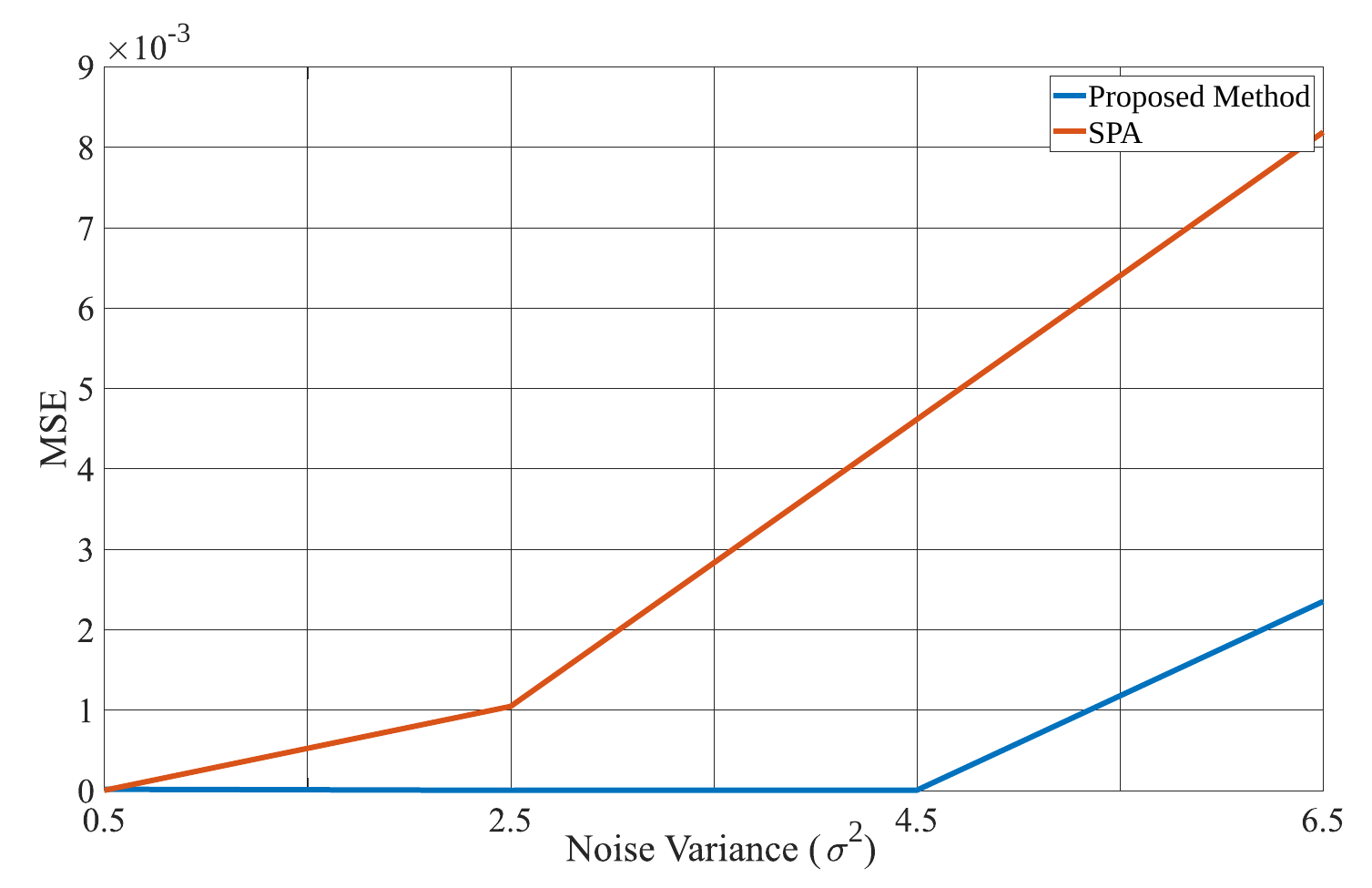}
    \caption{ Performance of the proposed method over SPA when the dense noise energy increases.}
    \label{fig7}
\end{figure}
 As seen, $\sigma^2$ can have a huge impact on the performance. Both methods tend to experience a threshold after which their MSEs increase dramatically. However, this threshold for the proposed method is significantly higher than the SPA, which indicates the greater robustness of the proposed method against Gaussian noise energy.
\section{Proof of Theorem 1}\label{section7}
In order to prove that problem (\ref{eqn7}) achieves exact demixing, we construct a trigonometric dual polynomial. Following the same line of \cite{fernandez2017demixing}, we apply the following kernel to build up the dual polynomial,
\begin{eqnarray}
\bar{K}(f)&:=& \mathcal{D}_{0.247m}(f)\mathcal{D}_{0.339m}(f)\mathcal{D}_{0.414m}(f) =\sum_{l=-m}^{m}c_le^{i2\pi lf},\nonumber\label{eqn16}
\end{eqnarray}
where ${N=2m+1}$, ${\boldsymbol c\in\mathbb{C}^N}$ is the convolution of the Fourier coefficients of the above kernels, and ${\mathcal{D}_m}$ is the Dirichlet kernel of order ${m>0}$ defined as
$${\mathcal{D}_m(f):=\frac{1}{N}\sum_{l=-m}^{m}e^{i2\pi lf}.}$$
According to the presence of outliers, conventional forms of dual polynomial can not be applied since the constraints (\ref{eqn14}) and (\ref{eqn13}) will not be met. Therefore, we use the randomized vector form of the dual polynomial presented in \cite{fernandez2017demixing} as
\begin{eqnarray}
\boldsymbol Q(f) &=& \boldsymbol Q_{aux}(f)+\boldsymbol R(f), \label{eqn161}
\end{eqnarray}
where
\begin{eqnarray}
\boldsymbol Q_{aux}(f) &=& \sum_{l\in \boldsymbol \Omega^c }^{ }\boldsymbol \Gamma_{l,:}e^{-i2\pi l f},\nonumber\label{eqn17} \\
\boldsymbol R(f) &=&\frac{1}{\sqrt{N}}\sum_{d\in \boldsymbol \Omega}^{}{  \boldsymbol r_{d,:}}e^{-i2\pi d f},\nonumber\label{eqn18}
\end{eqnarray}
where {  ${\boldsymbol r_{d,:}=\frac{\boldsymbol Z_{d,:}}{\|\boldsymbol Z_{d,:}\|_2}}$} and { ${\boldsymbol r\in\mathbb{C}^{s\times L}}$. }Note that (\ref{eqn13}) is immediately satisfied since ${\lambda=1/\sqrt{N}}$. Now we should build up the dual polynomial so that the other constraints in Lemma \ref{lemma1} are met. Using the same interpolation technique of \cite{fernandez2016super}, we set the value of the dual polynomial equal to ${\frac{c_k}{|c_k|}\boldsymbol b_k^{H}=h_k\boldsymbol b_k^{H}}$ at ${f_k\in\mathbb{T}}$ and set the derivative of the dual polynomial equal to zero at the same points. Setting the derivative to zero forces the dual polynomial to shape such that ${f_k}$ be a local extremum and bounds the value of the dual polynomial at these points. Thus, the following set of equations is formed for any ${f_k\in\mathbb{T}}$
\begin{subequations}
	\begin{eqnarray}
	\boldsymbol Q(f_k) &=& h_k\boldsymbol b_k^H, \label{eqn19}\\
	\boldsymbol Q_R^{(1)}(f_k)+i\boldsymbol Q_I^{(1)} &=& 0,\label{eqn20}
	\end{eqnarray}
\end{subequations}
where ${\boldsymbol Q_R^{(1)}}$ denotes the real part of the first derivative of ${\boldsymbol Q}$ and ${\boldsymbol Q_I}$ is the imaginary part of ${\boldsymbol Q}$. Using (\ref{eqn161}) in the above equations yields
\begin{subequations}
	\begin{align}
	&   \boldsymbol Q_{aux}(f_k)=h_k\boldsymbol b_k^H-\boldsymbol R(f_k),\label{eqn21}\\
	&  (\boldsymbol Q_{aux})_R^{(1)}(f_k)+i(\boldsymbol Q_{aux})_I^{(1)}(f_k)=-\boldsymbol R_R^{(1)}(f_k)-i\boldsymbol R_I^{(1)}(f_k).\label{eqn22}
	\end{align}
\end{subequations}
To interpolate ${\boldsymbol Q(f)}$ with ${\bar{K}(f)}$, we need to confine the kernel to ${\boldsymbol \Omega^c}$, as discussed for the missing data case in \cite{tang2013compressed}. Thus,
\begin{eqnarray}\label{eqn23}
K(f) &:=& \sum_{l\in\boldsymbol \Omega^c}^{}c_le^{i2\pi lf}=\sum_{l=-m}^{m}\delta_{\boldsymbol \Omega^c}(l)c_le^{i2\pi lf},
\end{eqnarray}
where ${\delta_{\boldsymbol \Omega^c}(l)}$ are Bernoulli random variables with parameter ${\frac{N-s}{N}}$. Therefore, ${\mathbb{E}K}$ is an scaled version of ${\bar{K}}$
\begin{eqnarray}
\mathbb{E}K(f) &=& \frac{N-s}{N}\sum_{l=-m}^{m}c_le^{i2\pi lf}=\frac{N-s}{N}\bar{K}(f).\label{eqn24}
\end{eqnarray}
The asymptotic behaviour of ${K(f)}$, ${\bar{K}(f)}$, and their derivatives is investigated in \cite{fernandez2016super}. With ${K(f)}$ restricted to ${\boldsymbol \Omega^c}$ we can express ${\boldsymbol Q_{aux}}$ in terms of ${K(f)}$ and its first derivative as
\begin{eqnarray}
\boldsymbol Q_{aux} &=& \sum_{k=1}^{K}\boldsymbol \alpha_kK(f-f_k)+\kappa\boldsymbol \beta_kK^{(1)}(f-f_k),\label{eqn25}
\end{eqnarray}
{ where ${\boldsymbol \alpha\in \mathbb{C}^{K\times L}}$ and ${\boldsymbol \beta\in \mathbb{C}^{K\times L}}$ are such that (\ref{eqn21}) and (\ref{eqn22}) are satisfied and ${\kappa:=1/\sqrt{\bar{K}^{(2)}(0)}}$. $\boldsymbol{\alpha}_k$ is the $k$th row of $\boldsymbol \alpha$ and $\boldsymbol{\beta}_k$ is the $k$th row of $\boldsymbol \beta$.}  The system of equations is then represented as
\begin{eqnarray}\label{eqn26}
\left[
\begin{array}{cc}
\boldsymbol D_0 & \boldsymbol D_1 \\
\boldsymbol D_1^T& \boldsymbol D_2 \\
\end{array}
\right]\left[
\begin{array}{c}
\boldsymbol \alpha \\
\boldsymbol \beta \\
\end{array}
\right]
&=& \left[
\begin{array}{c}
\boldsymbol \Phi \\
\boldsymbol 0 \\
\end{array}
\right]-\frac{1}{\sqrt{N}}\boldsymbol B_{\boldsymbol \Omega}\boldsymbol r,
\end{eqnarray}
where ${\boldsymbol 0\in\mathbb{C}^{K\times L}}$ is a zero matrix,${\boldsymbol \Phi_{k,:}=h_k\boldsymbol b_k^H}$,
\begin{eqnarray}
(D_0)_{jl} = K(f_j-f_l) ,(D_1)_{jl} = \kappa K^{(1)}(f_j-f_l) ,(D_2)_{jl} = -\kappa^2 K^{(2)}(f_j-f_l),\nonumber
\end{eqnarray}
$${\frac{1}{\sqrt{N}}\boldsymbol B_{\boldsymbol \Omega}\boldsymbol r=[\boldsymbol R(f_1),\ldots,\boldsymbol R(f_k),\boldsymbol R^{(1)}(f_1),\ldots
	,\boldsymbol R^{(1)}(f_k)]^T,}$$
$${\boldsymbol B_{\boldsymbol \Omega}=[\boldsymbol\nu(d_1),\ldots,\boldsymbol\nu(d_s)],}$$
\begin{eqnarray}
\boldsymbol\nu(g):=[e^{-i2\pi gf_1},\ldots &,& e^{-i2\pi gf_k}, \nonumber\\
i2\pi g\kappa e^{-i2\pi gf_1} &,& \ldots,i2\pi g \kappa e^{-i2\pi gf_k}]^T.  \nonumber
\end{eqnarray}
By solving (\ref{eqn26}), one can find ${\boldsymbol \alpha}$ and ${\boldsymbol \beta}$ and define ${\boldsymbol Q(f)}$ as
\begin{subequations}
	\begin{align}
	&\boldsymbol Q(f) = \sum_{k=1}^{K}\boldsymbol \alpha_kK(f-f_k)+\kappa\boldsymbol \beta_kK^{(1)}(f-f_k)+\boldsymbol R(f)\\\label{eqn27}
	&=\boldsymbol G_0^T(f)\boldsymbol D^{-1}\left(\left[
	\begin{array}{c}
	\boldsymbol \Phi \\
	\boldsymbol 0 \\
	\end{array}
	\right]-\frac{1}{\sqrt{N}}\boldsymbol B_{\boldsymbol \Omega}\boldsymbol r
	\right)+\boldsymbol R(f),
	\end{align}
\end{subequations}
where ${\boldsymbol G_p(f)}$ is defined as
\begin{align}
\boldsymbol G_p(f):=\kappa^p[K^{(p)}(f-f_1),\ldots,K^{(p)}(f-f_k),&\nonumber\\
\kappa K^{(p+1)}(f-f_1),\ldots,\kappa K^{(P+1)}(f-f_k)]^T&,
\end{align}
for ${p=0,1,2,\ldots}$ Now we should verify that the polynomial we formed above is guaranteed to be valid with high probability. If one can prove that ${\boldsymbol D^{-1}}$ exists, then (\ref{eqn26}) can be solved and (\ref{eqn11}) holds. Consider ${\bar{\boldsymbol D}}$ as the deterministic version of ${\boldsymbol D}$.  Lemma \ref{lem7} helps defining a condition under which ${\boldsymbol D^{-1}}$ exists and its deviation is bounded. We consider ${\varepsilon_{\boldsymbol D}^c}$ as the event in which ${\boldsymbol D^{-1}}$ exists with probability ${1-\epsilon/5}$ for ${\epsilon>0}$ under the assumption of Theorem \ref{theorem1}.
With this Lemma, one can conclude that in ${\varepsilon^c_{\boldsymbol D}}$ (\ref{eqn11}) holds. Note that (\ref{eqn13}) holds according to the definition of ${\boldsymbol Q(f)}$. All that remains is to prove (\ref{eqn12}) and (\ref{eqn14}). We use the results of Lemma \ref{lem6},\ref{lem8} below and Lemma 3.5 from \cite{fernandez2017demixing} which put bounds on the deviations of ${\boldsymbol B_{\boldsymbol \Omega}}$, ${\boldsymbol\nu(d)}$, and ${\boldsymbol G_p(f)}$, respectively. We use ${\varepsilon_{\boldsymbol B}^c}$ and ${\varepsilon_{\boldsymbol \nu}^c}$ as the events in which ${\boldsymbol B_{\boldsymbol \Omega}}$ and ${\boldsymbol \nu(d)}$ are bounded with probability at least ${1-\epsilon/5}$ under the assumption of Theorem \ref{theorem1}, respectively.
\begin{prop}\label{prop1}
	Under the assumption of Theorem \ref{theorem1} and conditioned on ${\varepsilon_{\boldsymbol B}^c\cap\varepsilon^c_{\boldsymbol D}\cap\varepsilon_{\boldsymbol \nu}^c}$, (\ref{eqn12}) holds with probability at least ${1-\epsilon/5}$.
\end{prop}
\begin{proof}\label{prop1_proof}
	Consider ${\bar{\boldsymbol Q}(f)}$ as the dual polynomial constructed using ${\bar{\boldsymbol K}(f)}$. We can rewrite (\ref{eqn27}) in a more general form for ${\boldsymbol K(f)}$ and ${\bar{\boldsymbol K}(f)}$ as
	\begin{align}
	&\kappa^{\iota}\bar{\boldsymbol Q}^{(\iota)}(f)  := \kappa^{\iota}\sum_{j=1}^{K}\bar{\boldsymbol\alpha}_j\bar{\boldsymbol K}^{(\iota)}(f-f_j)+ \nonumber\\
	&  \kappa^{\iota+1}\sum_{j=1}^{K}\bar{\boldsymbol\beta}_j\bar{\boldsymbol K}^{(\iota+1)}(f-f_j)  =\bar{\boldsymbol G}_{\iota}(f)^T\bar{\boldsymbol D}^{-1}\left[
	\begin{array}{c}
	\boldsymbol \Phi \\
	\boldsymbol 0 \\
	\end{array}
	\right]\label{eqn28}
	\\
	&\kappa^{\iota}{\boldsymbol Q}^{(\iota)}(f) := \kappa^{\iota}\sum_{j=1}^{K}{\boldsymbol\alpha}_j{\boldsymbol K}^{(\iota)}(f-f_j)+ \nonumber \\ &\kappa^{\iota+1}\sum_{j=1}^{K}{\boldsymbol\beta}_j{\boldsymbol K}^{(\iota+1)}(f-f_j)+\kappa^{\iota}\boldsymbol R^{(\iota)}(f) \nonumber \\
	& ={\boldsymbol G}_{\iota}(f)^T{\boldsymbol D}^{-1}\left(\left[
	\begin{array}{c}
	\boldsymbol \Phi \\
	\boldsymbol 0 \\
	\end{array}
	\right]-\frac{1}{\sqrt{N}}\boldsymbol B_{\boldsymbol \Omega}\boldsymbol r\right)+\kappa^{\iota}\boldsymbol R^{(\iota)}(f).\label{eqn29}
	\end{align}
	We can also express (\ref{eqn29}) as
	\begin{align}
	& \kappa^{\iota}{\boldsymbol Q}^{(\iota)}(f) :=\kappa^{\iota}\bar{\boldsymbol Q}^{(\iota)}(f)+\kappa^{\iota}{\boldsymbol R}^{(\iota)}(f)-\frac{1}{\sqrt{N}}\boldsymbol G_{\iota}(f)^T\boldsymbol D^{-1}\boldsymbol B_{\boldsymbol \Omega}\boldsymbol r \nonumber\\
	& +(\boldsymbol G_{\iota}(f)-\frac{N-s}{N}\bar{\boldsymbol G}_{\iota}(f))^T\boldsymbol D^{-1}\left[
	\begin{array}{c}
	\boldsymbol \Phi \\
	\boldsymbol 0 \\
	\end{array}
	\right] \nonumber\\
	& +\frac{N-s}{N}\bar{\boldsymbol G}_{\iota}(f)^T(\boldsymbol D^{-1}-\frac{N}{N-s}\bar{\boldsymbol D}^{-1})\left[
	\begin{array}{c}
	\boldsymbol \Phi \\
	\boldsymbol 0 \\
	\end{array}
	\right].\nonumber
	\end{align}
	Noting that ${\|\boldsymbol Q(f)\|_2\leq\|\bar{\boldsymbol Q}(f)\|_2+\|\boldsymbol Q(f)-\bar{\boldsymbol Q}(f)\|_2}$, for (\ref{eqn12}) to hold, we should have ${\|\bar{\boldsymbol Q}(f)\|_2+\|\boldsymbol Q(f)-\bar{\boldsymbol Q}(f)\|_2\leq1}$.
	The following lemmas complete the proof.
	\begin{lem}\label{lemma2}
		Under the assumptions of Proposition \ref{prop1}, ${\|\boldsymbol Q(f)-\bar{\boldsymbol Q}(f)\|_2\leq10^{-2}}$.
	\end{lem}
	\begin{lem}\label{lemma3}
		Under the assumptions of Proposition \ref{prop1}, ${\|\bar{\boldsymbol Q}(f)\|_2<0.99}$. Also
		\begin{align}
		&\frac{1}{2}\frac{d^2\|\boldsymbol Q(f)\|_2}{df^2} =\|\boldsymbol Q^{\prime}\|_2^2+Re\{\boldsymbol Q^{\prime \prime}\boldsymbol Q^H(d)\}<0\label{eqn30},\\
		&\forall f\in A_{near}:=\{f||f-f_j|\leq0.09\:for f_j\in \mathbb{T}\}.\nonumber
		\end{align}
	\end{lem}
The proof of the above lemmas appear in Section \ref{section9}.
\end{proof}
Now, we prove (\ref{eqn14}) as the last step to prove Theorem \ref{theorem1}.
\begin{prop}\label{prop2}
	Under the assumption of Theorem \ref{theorem1} and conditioned on ${\varepsilon_{\boldsymbol B}^c\cap\varepsilon^c_{\boldsymbol D}\cap\varepsilon_{\boldsymbol \nu}^c}$, (\ref{eqn14}) holds with probability at least ${1-\epsilon/5}$.
\end{prop}
\begin{proof}\label{prop2_proof}
	We can express ${\boldsymbol \Gamma_{l,:}}$ as
	
	\begin{align}
	& \boldsymbol \Gamma_{l,:}=\sum_{j=1}^{K}c_l\boldsymbol \alpha_je^{i2\pi lf_j}+i2\pi l\kappa\sum_{j=1}^{K}\boldsymbol \beta_je^{i2\pi lf_j} \nonumber\\
	& =c_l\boldsymbol\nu(l)^H\left[
	\begin{array}{c}
	\boldsymbol \alpha \\
	\boldsymbol \beta \\
	\end{array}
	\right]=c_l\boldsymbol\nu(l)^H\boldsymbol D^{-1}\left(\left[
	\begin{array}{c}
	\boldsymbol \Phi \\
	\boldsymbol 0 \\
	\end{array}
	\right]-\frac{1}{\sqrt{N}}\boldsymbol B_{\boldsymbol \Omega}\boldsymbol r\right)\nonumber\\
	&=c_l\left(<P\boldsymbol D^{-1}\boldsymbol \nu(l),\boldsymbol \Phi>+\frac{1}{\sqrt{N}}<\boldsymbol B_{\boldsymbol \Omega}^H\boldsymbol D^{-1}\boldsymbol \nu(l),\boldsymbol r>\right).\label{eqn31}
	\end{align}
	We use the results from \cite{fernandez2017demixing} to bound ${\|P\boldsymbol D^{-1}\boldsymbol \nu(l)\|_2}$ and ${\boldsymbol B_{\boldsymbol \Omega}^H\boldsymbol D^{-1}\boldsymbol \nu(l)}$,
	\begin{subequations}
		\begin{align}
		& \|P\boldsymbol D^{-1}\boldsymbol \nu(l)\|_2^2\leq640K\leq\frac{0.18^2N}{log40/\epsilon}\:\:\:in\:\: \varepsilon_D^c\label{eqn32}\\
		& \|\boldsymbol B_{\boldsymbol \Omega}^H\boldsymbol D^{-1}\boldsymbol \nu(l)\|_2^2\leq640C_{\boldsymbol B}^2KN\leq\frac{0.18^2N^2}{log40/\epsilon}\:\:\:in\:\:\varepsilon_D^c\cap\varepsilon_B^c.\label{eqn33}
		\end{align}
	\end{subequations}
	By applying the vector form of the Hoeffding's inequality \cite{yang2016exact} with ${t=0.18\sqrt{N}}$ for (\ref{eqn32}) and ${t=0.18N}$ for (\ref{eqn33}), we can conclude that each term in (\ref{eqn31}) is greater than its corresponding ${t}$ with probability ${\epsilon/10}$. Thus,
	\begin{align}
	& \|\boldsymbol \Gamma_{l,:}\|_{\infty,2}\leq\nonumber\\
	&\|\boldsymbol c\|_{\infty}\left(\|\boldsymbol\nu(l)^H\boldsymbol D^{-1}P^T\boldsymbol \Phi\|_2+\frac{1}{\sqrt{N}}\|\boldsymbol\nu(l)^H\boldsymbol D^{-1}\boldsymbol B_{\boldsymbol \Omega}\boldsymbol r\|_2\right) \nonumber\\
	& \leq\frac{2.6}{N}(0.36\sqrt{N})=\frac{0.936}{\sqrt{N}}<\frac{1}{\sqrt{N}},\nonumber
	\end{align}
	with probability at least ${1-\epsilon/5}$.
\end{proof}
\section{Conclusion and Future Work}\label{section8}
The problem of demixing exponential form signals and outliers using MMVs was discussed. A new convex optimization problem was proposed to solve the demixing problem. It was shown that with the minimum frequency separation condition satisfied, there exists a dual polynomial which interpolates the sign pattern of the signal and helps estimating the signal frequencies. Also, the dual variable was utilized to localize the outliers in the receiver. \par
As an extension to this work, one can investigate the demixing problem using an arbitrary sampling scheme. This is the case when integer sampling is not possible. Also, the computational complexity of the available SDPs is high. For practical purposes, it is mandatory to reduce the computational complexity of the proposed method.
\section{Proofs}\label{section9}
\subsection{Proof of Lemma\ref{lemma2}}
First, we bound ${\|\kappa^{\iota}\boldsymbol Q^{(\iota)}(f)-\kappa^{\iota}\bar{\boldsymbol Q}^{(\iota)}(f)\|_2}$ on a grid. Then, the result is extended to the continuous domain ${[0,1]}$ and then (\ref{eqn12}) is proved. In order to bound ${\|\kappa^{\iota}\boldsymbol Q^{(\iota)}(f)-\kappa^{\iota}\bar{\boldsymbol Q}^{(\iota)}(f)\|_2}$, we can bound each term in
\begin{align}
& \|\kappa^{\iota}{\boldsymbol R}^{(\iota)}(f)\|_2+\|\frac{1}{\sqrt{N}}\boldsymbol G_{\iota}(f)^T\boldsymbol D^{-1}\boldsymbol B_{\boldsymbol \Omega}\boldsymbol r\|_2 \nonumber\\
& +\|(\boldsymbol G_{\iota}(f)-\frac{N-s}{N}\bar{\boldsymbol G}_{\iota}(f))^T\boldsymbol D^{-1}\left[
\begin{array}{c}
\boldsymbol \Phi \\
\boldsymbol 0 \\
\end{array}
\right]\|_2 \nonumber\\
& +\|\frac{N-s}{N}\bar{\boldsymbol G}_{\iota}(f)^T(\boldsymbol D^{-1}-\frac{N}{N-s}\bar{\boldsymbol D}^{-1})\left[
\begin{array}{c}
\boldsymbol \Phi \\
\boldsymbol 0 \\
\end{array}
\right]\|_2\label{eqn34}
\end{align}
on a grid ${\mathcal{G}}$ such that ${|\mathcal{G}|=200\sqrt{L}N^3}$ where ${|\mathcal{G}|}$ is the cardinality of ${\mathcal{G}}$. Since, ${\iota\in\{0,1,2,3\}}$, we are dealing with ${|\mathcal{U}|=4|\mathcal{G}|}$ points. To bound each term in (\ref{eqn34}), we leverage Lemma 4 of \cite{yang2018sample}, which is stated as follows.

	\begin{lem}[\cite{yang2018sample} Lemma 4]\label{lem4}
		Consider a matrix ${\boldsymbol \Psi\in\mathbb{C}^{K\times L}}$ with rows $\boldsymbol \{\Psi_k\}_{k=1}^K$ and the vector $\boldsymbol{0}\neq \boldsymbol{\omega}\in \mathbb{C}^K$. If the rows of ${\boldsymbol \Psi}$ are independently distributed on the complex hyper-sphere ${\mathbb{S}^{2L-1}}$, then for all $t > \|\boldsymbol{w}\|_2$, we have
		\begin{align}
		\mathbb{P}\{\|\sum_{k=1}^K\boldsymbol \omega_k \boldsymbol\Psi_k\|_2\geq t\}\leq e^{-L\left( \frac{t^2}{\|\omega\|_2^2}-log \frac{t^2}{\|\omega\|_2^2} -1\right)} & \:\forall\boldsymbol \omega\in\mathbb{C}^K,\boldsymbol\omega\neq0,t>0.\label{eqn35}
		\end{align}
	\end{lem}
	Each term in (\ref{eqn34}) is associated with an event ${\varepsilon_q}$ and ${q=\{1,2,3,4\}}$. For the ease of reading, we separate the proof of the bounds on each term.
{ 	\subsubsection{Bound on ${\varepsilon_1}$}}
	 The first term in (\ref{eqn34}) can be expressed as
	\begin{align}
	\kappa^{\iota}{\boldsymbol R}^{(\iota)}(f) & = \frac{\kappa^{\iota}}{\sqrt{N}}\sum_{d\in\boldsymbol \Omega}^{ }\boldsymbol r_{d,:}(i2\pi d)^{(\iota)}e^{-i2\pi df}\:\:\:(\iota)=\{0,1,2,3\}.\nonumber
	\end{align}
	Therefore, we define
	\begin{align}
	\varepsilon_1 & :=\{\|\kappa^{\iota}{\boldsymbol R}^{(\iota)}(f)\|_2\geq t\:\:\: for\:\: all\:\: f\in  |\mathbb{T}_{grid}|\}.\nonumber
	\end{align}
	By setting ${\boldsymbol\Psi=\boldsymbol r}$ and
	\begin{align}
	\boldsymbol \omega & =\frac{\kappa^{\iota}}{\sqrt{N}}\left[(i2\pi l_1)^{(\iota)}e^{i2\pi l_1f},\ldots,(i2\pi l_s)^{\iota}e^{i2\pi l_sf}\right]^T,\nonumber
	\end{align}
	in (\ref{eqn35})
	and using the union bound, we can conclude that
	\begin{align}
	\mathbb{P}\{\underset{f\in\mathcal{U}}{sup}\|\kappa^{\iota}\boldsymbol R^{(\iota)}(f)\|_2\geq t\} & \leq |\mathbb{T}_{grid}|e^{-L\left( \frac{t^2}{\|\omega\|_2^2}-\log \frac{t^2}{\|\omega\|_2^2} -1\right)}.\label{eqn351}
	\end{align}
	If we set
	\begin{align}\label{eqn352}
	\frac{t^2}{\|\omega\|_2^2}-\log \frac{t^2}{\|\omega\|_2^2} -1\geq \frac{1}{L}\log \frac{|\mathbb{T}_{grid}|}{\epsilon/20},
	\end{align}
	we get at most $\epsilon/20$ probability of occurrence for (\ref{eqn351}). By leveraging \cite[Lemma 5]{yang2018sample},  a sufficient condition for (\ref{eqn352}) to hold is
	\begin{align}
	&\frac{t^2}{\|\omega\|_2^2}\geq 2(1+\frac{1}{L}\log \frac{|\mathbb{T}_{grid}|}{\epsilon/20})\nonumber\\
	&\rightarrow \|\omega\|_2^2\leq t^2\left( 2(1+\frac{1}{L}\log \frac{|\mathbb{T}_{grid}|}{\epsilon/20}) \right)^{-1}\leq \frac{t^2}{2}\left(1+\frac{1}{L}\log \frac{|\mathbb{T}_{grid}|}{\epsilon/20})\right)^{-1}.
	\end{align}
	The above result combined with the bound \cite{fernandez2017demixing}
	\begin{align}
	\|\boldsymbol \omega\|_2^2 & \leq\frac{\kappa^{2\iota}}{N}(2\pi m)^{2\iota}s\leq\frac{\pi^6s}{N},
	\end{align}
	leads to the sufficient condition,
	$$s\leq \frac{N}{\pi^6} (1+\frac{1}{L} \log\frac{|\mathbb{T}_{grid}|}{\epsilon})^{-1},$$ which is actually satisfied by the second sufficient condition in Theorem \ref{theorem1} after setting ${t=\frac{10^{-2}}{8}}$ and $C_s$ small enough. Thus, one can conclude that the event ${\varepsilon_1}$ happens with probability at most ${\epsilon/20}$ under the assumptions of Proposition \ref{prop1}.
{ 	\subsubsection{Bound on ${\varepsilon_2}$}}
Following the same procedure as for ${\varepsilon_1}$, one can bound the second term in (\ref{eqn34}). Consider ${\boldsymbol \Psi=\boldsymbol r}$ and
\begin{align}
\boldsymbol \omega & =\frac{1}{\sqrt{N}}\boldsymbol G_{\iota}^T(f)\boldsymbol D^{-1}\boldsymbol B_{\boldsymbol \Omega}.\nonumber
\end{align}
Note that we can write
\begin{align}
\|\frac{1}{\sqrt{N}}\boldsymbol G_{\iota}^T(f)\boldsymbol D^{-1}\boldsymbol B_{\boldsymbol \Omega}\|_2 & \leq\frac{1}{\sqrt{N}}\|\boldsymbol B_{\boldsymbol \Omega}\|\|\boldsymbol D^{-1}\|\|\boldsymbol G_{\iota}(f)\|_2,\label{eqn36}
\end{align}
{ where $\|.\|$ denotes the operator norm.} Now, we should find the sufficient conditions for bounding each term of (\ref{eqn36}). {The bound for the terms $\|\boldsymbol D^{-1}\|$ and $\|\boldsymbol G_{\iota}(f)\|_2$ can be found below in Lemmas \ref{lem7} and \ref{lem8}, respectively.} Lemma \ref{lem5} will provide a new bound for $\|\boldsymbol B_{\boldsymbol \Omega}\|$.
\begin{lem}\label{lem5}
Under the assumptions of Theorem \ref{theorem1}, the event
\begin{align}
    \varepsilon_{\boldsymbol B}=\left\{ \|\boldsymbol B_{\boldsymbol \Omega}\|>C_{\boldsymbol B}\sqrt{N}\left( \log(\frac{N}{\epsilon})(1+\frac{1}{L}\log\frac{N^3\sqrt{L}}{\epsilon}) \right)^{-1/2}  \right\}\nonumber
\end{align}
will occur with probability at most $\epsilon/5$ for some constant $C_{\boldsymbol B}$.
\end{lem}
\begin{proof}
    Define $H:=\boldsymbol B_{\boldsymbol \Omega}\boldsymbol B_{\boldsymbol \Omega}^H$ which is
    $$H=\sum_{l\in\boldsymbol\Omega}\nu(l)\nu^*(l).$$
    The matrix $H$ is dissipated around $\bar{H}=\sum_{l=-m}^m\nu(l)\nu^*(l)$. Using the result of Lemma E.1 in \cite{fernandez2017demixing}, we have
    \begin{align}\label{eqn39_new}
        \|\bar{H}\|\leq 260\pi^2 N \log K.
    \end{align}
    Using the bound on $s$ from Theorem \ref{theorem1}, we can write
    $$s\leq C_sN(\bigg(\log\frac{N}{\epsilon}\bigg)^{-1}\bigg(1+\frac{1}{L}\log(\frac{\sqrt{L}N^3}{\epsilon})\bigg)^{-1}\leq  C_sN\big(\log K\big)^{-1}\bigg(1+\frac{1}{L}\log(\frac{\sqrt{L}N^3}{\epsilon})\bigg)^{-1}.$$
    Then from (\ref{eqn39_new}) and the above bound, we can bound $\|\frac{s}{N}\bar{H}\|$ as
    \begin{align}
        \|\frac{s}{N}\bar{H}\|\leq \frac{260\pi^2C_sN}{1+\frac{1}{L}\log\frac{\sqrt{L}N^3}{\epsilon}}=\frac{C_{\boldsymbol B}^2}{2}N(1+\frac{1}{L}\log\frac{\sqrt{L}N^3}{\epsilon})^{-1}.\nonumber
    \end{align}
    Now, we can control the deviation of $H$ from $\bar{H}$ using the following Lemma.
    \begin{lem}\label{lem6}
    Under the assumptions of Theorem \ref{theorem1}
    $$\|H-\frac{s}{N}\bar{H}\|\leq \frac{C_{\boldsymbol B}^2}{2}N(1+\frac{1}{L}\log\frac{\sqrt{L}N^3}{\epsilon})^{-1}$$
    with probability at least $1-\frac{\epsilon}{5}$.

    \end{lem}
    The proof of the above Lemma is given after the current proof. Using the result from Lemma \ref{lem6}, we have
    \begin{align}
    \|\boldsymbol B_{\boldsymbol \Omega}\|&\leq \sqrt{\|H\|}\leq \sqrt{\frac{s}{N}\|\bar{H}\|+\|H-\frac{s}{N}\bar{H}\|}\nonumber\\
    &\leq \sqrt{\frac{C_{\boldsymbol B}^2}{2}N(1+\frac{1}{L}\log\frac{\sqrt{L}N^3}{\epsilon})^{-1}+\frac{C_{\boldsymbol B}^2}{2}N(1+\frac{1}{L}\log\frac{\sqrt{L}N^3}{\epsilon})^{-1}}\nonumber\\
    &=C_{\boldsymbol B}\sqrt{N}(1+\frac{1}{L}\log\frac{\sqrt{L}N^3}{\epsilon})^{-1/2}
    \end{align}
     with probability at least $1-\frac{\epsilon}{5}$. This concludes the proof of Lemma \ref{lem5}.
\end{proof}
\begin{proof}[Proof of Lemma \ref{lem6}]
    Under the assumptions of Theorem \ref{theorem1}, one can write
     $$H=\sum_{l=-m}^m\delta_{\Omega}(l)\nu(l)\nu^*(l),$$
     where $\delta_{\Omega}(l),l=-m,...,m$ are i.i.d. Bernoulli random variables with parameter $\frac{s}{n}$. Next we can build zero-mean self adjoint matrices from $H$ as
     $$X_l:=(\delta_{\Omega}-\frac{s}{N})\nu(l)\nu^*(l),$$
     so that we can apply Matrix Bernstein inequality \cite{Tropp2012UserFriendlyTB}.
     \begin{thm}[Matrix Bernstein inequality \cite{Tropp2012UserFriendlyTB}]\label{thm2}
         Let $\{X_l\}$ be a finite sequence of independent zero-mean self-adjoint random matrices of dimension $d$ such that $\|X_l\|\leq B$ almost surely for a certain constant $B$. For all $t\geq 0$ and a positive constant $\sigma^2$
         \begin{align}\label{eqn41_new}
             \mathbb{P}\{\|\sum_{l=-m}^mX_l\|\geq t\}\leq de^{-\frac{t^2/2}{\sigma^2+Bt/3}}
         \end{align}
         for $\|\sum_{l=-m}^m\mathbb{E}(X_l^2)\|\leq \sigma^2$.
     \end{thm}
     In order to be able to apply the recent theorem on $X_l$, we need a bound on $\|X_l\|$. Using Lemma 3.5 in \cite{fernandez2017demixing}, we have
     $$\|X_l\|\leq \sup_{-m\leq l\leq m}\|\nu(l)\|_2^2\leq B:=10K.$$
     Also, to find the value for $\sigma^2$, we can write
     \begin{align}
        \sigma^2 &:=\|\sum_{l=-m}^m\mathbb{E}(X_l^2)\|=\|\sum_{l=-m}^m\mathbb{E}((\bar{\delta}(l)-\frac{s}{N}))^2\|\nu(l)\|_2^2\nu(l)\nu^*(l)\|\nonumber\\
        &\leq 10K\frac{s}{N}\|\bar{H}\|\leq 10C_{\boldsymbol B}^2NK(\log\frac{N}{\epsilon}(1+\frac{1}{L}\log\frac{\sqrt{L}N^3}{\epsilon}))^{-1}.\nonumber
     \end{align}
     Thus, if we set $t:=\frac{C_{\boldsymbol B}^2N}{2}(\log\frac{N}{\epsilon}(1+\frac{1}{L}\log\frac{\sqrt{L}N^3}{\epsilon}))^{-1}$, we can take $\sigma^2=20Kt$ in Theorem \ref{thm2}. This will yield
     $$\mathbb{P}\{ \|H-\frac{s}{n}\bar{H}\|\geq t \}\leq 2Ke^{\frac{-t^2/2}{\sigma^2+Bt/3}}=2Ke^{\frac{-3t}{140K}}.$$
     The above inequality will lead to the conclusion that to get the maximum probability of occurrence of $\epsilon/5$, we should have
     $$K\leq \frac{3C_{\boldsymbol B}^2N}{280}(\log\frac{10K}{\epsilon}(1+\frac{1}{L}\log\frac{N^3\sqrt{L}}{\epsilon}))^{-1},$$
     which is satisfied by the bound on $K$ in Theorem \ref{theorem1} if we set $C_K$ small enough.
\end{proof}

{
\begin{lem}\label{lem7}
Under the assumptions of Theorem \ref{theorem1}, the event,
\begin{align}
    \varepsilon_D :=\left\{ \|D-\frac{N-s}{N}\bar{D}\|\geq \frac{N-s}{4N} C_D\left( 1+\frac{1}{L}\log(\frac{|\mathbb{T}_{grid}|}{\epsilon}) \right)^{-\frac{1}{2}}\right\}\nonumber
\end{align}
occurs with probability at most $\epsilon/5$. Also, in the complement event $\varepsilon_D^c$, $D^{-1}$ exists and
$$\|D^{-1}\|\leq 8,\space \|D^{-1}-\frac{N}{N-s}\bar{D}^{-1}\|\leq C_D\left(  1+\frac{1}{L}\log(\frac{|\mathbb{T}_{grid}|}{\epsilon}) \right)^{-\frac{1}{2}},$$
where $C_D$ is a constant.
\end{lem}
\begin{proof}
    In order to prove, we use lemmas G.1 and G.3 from \cite{fernandez2017demixing}, which both hold in our setting. Next, using Lemma G.1 and triangle inequality, we can bound the smallest singular value of $D$ by a positive number which ensures invertibility of $D$. Then, by setting $A = D$ and $B=p\bar{D}$ in Lemma G.3, we get $\|D^{-1}\|\leq 8$ (note that $N/(N-s)\leq 2$ with $s\leq N/2)$. Next, we define
    $$X_l:= (p-\delta_{\Omega^c}(l))c_l\nu (l)\nu^* (l)$$
   for any $-m\leq l\leq m$ with $p=(N-s)/N$. Note that $\bar{D}=\sum_{l=-m}^m c_l\nu(l)\nu^*(l)$ and $D=\sum_{l=-m}^m\delta_{\Omega^c}(l)c_l\nu(l)\nu^*(l)$. Thus, $\mathbb{E}(X_l)=0$. Using the same calculation in \cite{fernandez2017demixing}(Lemmas 3.4, 3.5), $\|X_l\|\leq B:=\frac{12.6 K}{m}$. Also,
   $$\mathbb{E}(X_l^2)=p(1-p)c_l^2\|\nu(l)\|_2^2\nu(l)\nu^*(l),$$
   which leads to
   $\sum_{l=-m}^m \mathbb{E}(X_l^2)\leq \sigma^2 := \frac{18.5pk}{m}$
   as in \cite{fernandez2017demixing}. Now, set $t= \frac{p}{4}C_D'(1+\frac{1}{L}\log \frac{|\mathbb{T}_{grid}|}{\epsilon})^{-\frac{1}{2}}$ with $C_D'=\min\{1,C_D/4\}$ in  Theorem \ref{thm2} with $B$ and $\sigma ^2$ as defined above. Then, we get
   \begin{align}
       \frac{t^2/2}{\sigma^2+Bt/3}&=\frac{pm}{32K}C_D'^2\left[ 18.5(1+\frac{1}{L}\log \frac{|\mathbb{T}_{grid}|}{\epsilon}) + 1.05 C_D'(1+\frac{1}{L}\log \frac{|\mathbb{T}_{grid}|}{\epsilon})^{\frac{1}{2}}\right]^{-1}\nonumber\\
       &> \frac{pm}{32K}C_D'\left( 1+\frac{1}{L}\log(\frac{|\mathbb{T}_{grid}|}{\epsilon}) \right)^{-1}[18.05+1.05C_D']^{-1}\nonumber\\
       & = \frac{N-s}{K}C_D''\left( 1+\frac{1}{L}\log(\frac{|\mathbb{T}_{grid}|}{\epsilon}) \right)^{-1}.\nonumber
   \end{align}
   Theorem \ref{thm2} implies that
   $$\mathbb{P}\{\|D-p\bar{D}\|>t\}\leq 2Ke^{-\frac{N-s}{K}C_D''\left( 1+\frac{1}{L}\log(\frac{|\mathbb{T}_{grid}|}{\epsilon}) \right)^{-1}}.$$
   This probability is smaller than $\epsilon/5$ as long as
   $$K< NC_D''\left( 1+\frac{1}{L}\log(\frac{|\mathbb{T}_{grid}|}{\epsilon}) \right)^{-1}\left(  \log\frac{10K}{\epsilon} \right)^{-1}, \space s<N/2,$$
   which is the case by assumptions in Theorem \ref{theorem1} if $C_K$ and $C_s$ are small enough. This concludes the proof.
\end{proof}

\begin{lem}\label{lem8}
Consider the equispaced grid $\mathcal{G}\subset [0,1]$ with cardinality $|\mathbb{T}_{grid}|=200\sqrt{L}N^3$. Then, the event
\begin{align}
    \varepsilon_G :=\left\{ \|G_{\iota}(f)-\frac{N-s}{N}\bar{G}_{\iota}(f)\|_2>  C_G\left( 1+\frac{1}{L}\log(\frac{|\mathbb{T}_{grid}|}{\epsilon}) \right)^{-\frac{1}{2}}\right\}\nonumber
\end{align}
for any $f\in \mathcal{G}$, $\iota\in \{0,1,2,3\}$, and constant $C_G$, has probability bounded by $\epsilon/5$.
\end{lem}
\begin{proof}
    We need the vector Bernstein inequality to prove this lemma.
    \begin{thm}[Vector Bernstein inequality \cite{candes2011probabilistic}]\label{theorem3}
        Let $\mathcal{P}\subset \mathbb{R}^d$ be a finite sequence of independent zero-mean random vectors with $\|\boldsymbol p\|_2\leq B \quad a.s.$ and $\sum_{\boldsymbol p\in \mathcal{P}}\mathbb{E} \|\boldsymbol p\|_2^2\leq \sigma^2$ for all $\boldsymbol p\in \mathcal{P}$, where $B$ and $\sigma^2$ are both positive constants. Then,
        $$\mathbb{P}\left\{ \|\sum_{\boldsymbol p\in\mathcal{P}}\boldsymbol p\|_2\geq t  \right\}\leq e^{-\frac{t^2}{8\sigma^2}+\frac{1}{4}}$$
        for $0\leq t\leq \frac{\sigma^2}{B}$.
    \end{thm}

Using the definition of $K$ and $\bar{K}$, we can respectively rewrite $\boldsymbol G_{\iota}(f)$ and $\bar{\boldsymbol G}_{\iota}(f)$ as
$$\boldsymbol G_{\iota}(f) = \sum_{l=-m}^m \delta_{\Omega^c}(l)(i2\pi\kappa l)^{\iota}c_le^{i2\pi l f}\boldsymbol \nu(l) ,\bar{G}_{\iota}(f)=\sum_{l=-m}^m (i2\pi\kappa l)^{\iota}c_le^{i2\pi l f}\boldsymbol \nu(l).$$
Note that by defining
$$\boldsymbol p(\iota,l):= (\delta_{\Omega^c}(l)-p)(i2\pi \kappa l)^{\iota}c_le^{i2\pi lf}\boldsymbol \nu(l),$$
where $p=\frac{N-s}{N}$ (parameter of i.i.d Bernoulli random variables $\delta_{\Omega^c}(-m),...,\delta_{\Omega^c}(m)$) we have
$\boldsymbol G_{\iota}(f)-p\bar{\boldsymbol G}_{\iota}(f)=\sum _{l=-m}^m \boldsymbol p(l)$. Also, using Lemmas 3.3, 3.4, and 3.5 from \cite{fernandez2017demixing}, we get
$$\|\boldsymbol p (\iota,l)\|_2\leq B:= \frac{128\sqrt{K}}{m},\sum_{l=-m}^m\mathbb{E}(\| \boldsymbol p(\iota,l)\|_2^2)\leq \sigma^2:=\frac{3.25 10^4K}{m}.$$
Now, using Theorem \ref{theorem3}, we have
$$\mathbb{P}\left\{  \sup_{f\in\mathcal{G}}\|\boldsymbol G_{\iota}(f)-p\bar{\boldsymbol G}_{\iota}(f) \|_2\geq t ,\quad \iota=\{0,1,2,3\}\right\}\leq 4|\mathbb{T}_{grid}|e^{-\frac{t^2}{8\sigma^2}+\frac{1}{4}}.$$
To make the r.h.s. smaller than $\epsilon/5$, take
$$t:= \sqrt{\frac{26\times 10^4 K}{m}(\frac{1}{4}+\log(\frac{20|\mathbb{T}_{grid}|}{\epsilon}))},$$
which is a valid choice since
\begin{align}
    \frac{t}{\sigma}&=\sqrt{8(\frac{1}{4}+\log(\frac{20*|\mathbb{T}_{grid}|}{\epsilon}))}\leq \sqrt{74+24\log(N)+4\log(L)+8\log(\frac{1}{\epsilon})}\nonumber\\
    &\leq \sqrt{74+44\log(N)}+\sqrt{8\log(\frac{1}{\epsilon})}
  \leq 0.452\sqrt{N}\sqrt{8\log(\frac{1}{\epsilon})}\leq 0.46\sqrt{N},\nonumber
\end{align}
where we have used $\sqrt{74+44\log(N)}\leq 0.452\sqrt{N}$ and assumed $N\geq 2 \times 10^3$, $L\leq N^5$, and either $K\geq 1$ or $s\geq 1$. Thus, $t/\sigma\leq 0.46\sqrt{N}\leq \sigma/B$. The desired result holds as long as
$$C_G(1+ \frac{1}{L}\log(\frac{|\mathbb{T}_{grid}|}{\epsilon}))^{-\frac{1}{2}} \geq t\geq \sqrt{\frac{2\times 10^3 K}{N}(\frac{1}{4}+\log(\frac{8\times 10^3 \sqrt{L}N^3}{\epsilon}))}$$
with $C_K$ small enough.
\end{proof}
}

Using  Lemmas \ref{lem7}, H.8 and Corollary H.9 from \cite{fernandez2017demixing} with respect to the new bounds for $s$ and $K$ and the first condition of Theorem \ref{theorem1} combined with Lemma \ref{lem5}, we find tight bounds for (\ref{eqn36}) as
\begin{align}
&  \frac{1}{\sqrt{N}}\|\boldsymbol B_{\boldsymbol \Omega}\|\|\boldsymbol D^{-1}\|\|\boldsymbol G_{\iota}(f)\|_2 \leq\frac{8(C_{\bar{\boldsymbol \nu}}+C_{{\boldsymbol \nu}})\|\boldsymbol B_{\boldsymbol \Omega}\|}{\sqrt{N}} \nonumber\\
&  \leq\frac{8(C_{\bar{\boldsymbol \nu}}+C_{{\boldsymbol \nu}})C_{\boldsymbol B}\left(1+\frac{1}{L}\log \frac{|\mathbb{T}_{grid}|}{\epsilon}\right)^{-\frac{1}{2}}\sqrt{N}}{\sqrt{N}}\:\:\:,C_{\boldsymbol B}=\frac{C_{\mathcal{U}}}{8\left(C_{\bar{\boldsymbol \nu}}+C_{{\boldsymbol \nu}}\right)}\nonumber.
\end{align}
Thus, by setting ${t=\frac{10^{-2}}{8}}$ and using Lemma \ref{lem4} and the union bound, we obtain
\begin{align}
&\mathbb{P}\{\underset{f\in\mathcal{U}}{sup}\|\frac{1}{\sqrt{N}}\boldsymbol G_{\iota}^T(f)\boldsymbol D^{-1}\boldsymbol B_{\boldsymbol \Omega}\|_2\geq \frac{10^{-2}}{8}\}\nonumber\\ &\leq |\mathbb{T}_{grid}|e^{-L\left(\frac{{\frac{10^{-2}}{8}}^2}{C_{\mathcal{U}}^2(1+\frac{1}{L}\log\frac{|\mathbb{T}_{grid}|}{\epsilon})^{-1}}-\log\frac{{\frac{10^{-2}}{8}}^2}{C_{\mathcal{U}}^2(1+\frac{1}{L}\log\frac{|\mathbb{T}_{grid}|}{\epsilon})^{-1}}-1\right)}.\label{eqn37}
\end{align}
With the same reasoning for $\varepsilon_1$ and small enough $C_{\mathcal{U}}$, the event
\begin{align}
\varepsilon_2 & :=\{\|\frac{1}{\sqrt{N}}\boldsymbol G_{\iota}^T(f)\boldsymbol D^{-1}\boldsymbol B_{\boldsymbol \Omega}\|_2\geq \frac{10^{-2}}{8}\:\:\: for\:\: all\:\: f\in \mathcal{U}\}\nonumber
\end{align}
holds with probability at most ${\epsilon/20}$ under the assumptions of Proposition\ref{prop1}.
{ \subsubsection{Bound on $\varepsilon_3$}}
For the third term, we consider ${\boldsymbol \Psi=\boldsymbol \Phi}$ and
\begin{align}
\boldsymbol \omega & =\boldsymbol P\boldsymbol D^{-1}\left(\boldsymbol G_{\iota}(f)-\frac{N-s}{N}\bar{\boldsymbol G}_{\iota}(f)\right),\nonumber
\end{align}
where ${\boldsymbol P\in\mathbb{R}^{K\times 2K}}$ is a projection matrix, which selects the first ${K}$ elements in a vector and ${\|\boldsymbol P\|=1}$. According to Lemmas \ref{lem7} and \ref{lem8} we can write
\begin{align}
&  \|\boldsymbol P\boldsymbol D^{-1}\left(\boldsymbol G_{\iota}(f)-\frac{N-s}{N}\bar{\boldsymbol G}_{\iota}(f)\right)\|_2 \leq \|\boldsymbol P\|\|\boldsymbol D^{-1}\|\|\boldsymbol G_{\iota}(f)-\frac{N-s}{N}\bar{\boldsymbol G}_{\iota}(f)\|_2\nonumber\\
& \leq 8\|\boldsymbol G_{\iota}(f)-\frac{N-s}{N}\bar{\boldsymbol G}_{\iota}(f)\|_2\leq C_{\mathcal{U}}\left(1+\frac{1}{L}\log\frac{|\mathbb{T}_{grid}|}{\epsilon}\right)^{-\frac{1}{2}}.\nonumber
\end{align}
By setting ${t=\frac{10^{-2}}{8}}$ and applying (\ref{eqn35}) and the union bound, we have
\begin{align}
&\mathbb{P}\{\underset{f\in\mathcal{U}}{sup}\|(\boldsymbol G_{\iota}(f)-\frac{N-s}{N}\bar{\boldsymbol G}_{\iota}(f))^T\boldsymbol D^{-1}\left[
\begin{array}{c}
\boldsymbol \Phi \\
\boldsymbol 0 \\
\end{array}
\right]\|_2\geq \frac{10^{-2}}{8}\}\nonumber\\ &\leq |\mathbb{T}_{grid}|e^{-L\left(\frac{{\frac{10^{-2}}{8}}^2}{C_{\mathcal{U}}^2(1+\frac{1}{L}\log\frac{|\mathbb{T}_{grid}|}{\epsilon})^{-1}}-\log\frac{{\frac{10^{-2}}{8}}^2}{C_{\mathcal{U}}^2(1+\frac{1}{L}\log\frac{|\mathbb{T}_{grid}|}{\epsilon})^{-1}}-1\right)}.\label{eqn38}
\end{align}
Therefore, with the same reasoning for $\varepsilon_1$ and $\varepsilon_2$, the event
\begin{align}
\varepsilon_3  :=\{\|(\boldsymbol G_{\iota}(f)-\frac{N-s}{N}\bar{\boldsymbol G}_{\iota}(f))^T\boldsymbol D^{-1}&\left[
\begin{array}{c}
\boldsymbol \Phi \\
\boldsymbol 0 \\
\end{array}
\right]\|_2\geq \frac{10^{-2}}{8}\:\:\: for\:\: all\:\: f\in \mathcal{U}\}\nonumber
\end{align}
holds with probability at most ${\epsilon/20}$ under the assumptions of Proposition\ref{prop1}.
{ \subsubsection{Bound on $\varepsilon_4$}}
 At last, one can bound the fourth term in (\ref{eqn34}) by considering ${\boldsymbol \Psi=\boldsymbol \Phi}$ and
\begin{align}
& \boldsymbol \omega=\frac{N-s}{N}\boldsymbol P\left(\boldsymbol D^{-1}-\frac{N}{N-s}\bar{\boldsymbol D}^{-1}\right)\bar{\boldsymbol G}_{\iota}(f).\nonumber
\end{align}
Using the  Lemmas \ref{lem7},\ref{lem8} we get
\begin{align}
& \|\frac{N-s}{N}\boldsymbol P\left(\boldsymbol D^{-1}-\frac{N}{N-s}\bar{\boldsymbol D}^{-1}\right)\bar{\boldsymbol G}_{\iota}(f)\|_2\leq \nonumber\\
& \|\boldsymbol P\|\|\boldsymbol D^{-1}-\frac{N}{N-s}\bar{\boldsymbol D}^{-1}\|\|\bar{\boldsymbol G}_{\iota}(f)\|_2\leq \nonumber\\
& C_{\bar{\boldsymbol \nu}}\|\boldsymbol D^{-1}-\frac{N}{N-s}\bar{\boldsymbol D}^{-1}\|\leq C_{\mathcal{U}}\left(1+\frac{1}{L}\log\frac{|\mathbb{T}_{grid}|}{\epsilon}\right)^{-\frac{1}{2}}.\nonumber
\end{align}
By applying (\ref{eqn35}) and the union bound and setting ${t=\frac{10^{-2}}{8}}$, one can write
\begin{align}
&\mathbb{P}\{\underset{f\in\mathcal{U}}{sup}\|\frac{N-s}{N}\bar{\boldsymbol G}_{\iota}(f)^T(\boldsymbol D^{-1}-\frac{N}{N-s}\bar{\boldsymbol D}^{-1})\left[
\begin{array}{c}
\boldsymbol \Phi \\
\boldsymbol 0 \\
\end{array}
\right]\|_2\geq \frac{10^{-2}}{8}\}\nonumber\\ &\leq |\mathbb{T}_{grid}|e^{-L\left(\frac{{\frac{10^{-2}}{8}}^2}{C_{\mathcal{U}}^2(1+\frac{1}{L}\log\frac{|\mathbb{T}_{grid}|}{\epsilon})^{-1}}-\log\frac{{\frac{10^{-2}}{8}}^2}{C_{\mathcal{U}}^2(1+\frac{1}{L}\log\frac{|\mathbb{T}_{grid}|}{\epsilon})^{-1}}-1\right)}.\label{eqn39}
\end{align}
Therefore, with the same reasoning for $\varepsilon_1$, $\varepsilon_2$, and $\varepsilon_3$ the event
\begin{align}
\varepsilon_4  :=\{\|\frac{N-s}{N}\bar{\boldsymbol G}_{\iota}(f)^T(\boldsymbol D^{-1}-\frac{N}{N-s}\bar{\boldsymbol D}^{-1})&\left[
\begin{array}{c}
\boldsymbol \Phi \\
\boldsymbol 0 \\
\end{array}
\right]\|_2\geq \frac{10^{-2}}{8}\nonumber\\
&\:\:\: for\:\: all\:\: f\in \mathcal{U}\}\nonumber
\end{align}
holds with probability at most ${\epsilon/20}$ under the assumptions of Proposition\ref{prop1}. Thus, using (\ref{eqn351}),(\ref{eqn37}),(\ref{eqn38}),(\ref{eqn39}), and the triangle inequality, we conclude that
\begin{align}\label{eqn40}
& \underset{f\in\mathcal{U}}{sup}\|\kappa^{\iota}\boldsymbol Q^{(\iota)}(f)-\kappa^{\iota}\bar{\boldsymbol Q}^{(\iota)}(f)\|_2\leq\frac{10^{-2}}{2}
\end{align}
holds with probability at least ${1-\epsilon/5}$ under the assumptions of Proposition\ref{prop1}. Next, using Bernstein polynomial inequality\cite{schaeffer1941inequalities}, we extend the results to the continuous domain ${[0,1]}$. Considering ${f\in[0,1]}$ and ${f_g\in\mathcal{G}}$, we have
\begin{align}
&  \|\kappa^{\iota}\boldsymbol Q^{(\iota)}(f)-\kappa^{\iota}\bar{\boldsymbol Q}^{(\iota)}(f)\|_2\leq\|\kappa^{\iota}\boldsymbol Q^{(\iota)}(f_g)-\kappa^{\iota}{\boldsymbol Q}^{(\iota)}(f)\|_2 \nonumber\\
&+\|\kappa^{\iota}\bar{\boldsymbol Q}^{(\iota)}(f_g)-\kappa^{\iota}{\boldsymbol Q}^{(\iota)}(f_g)\|_2+\|\kappa^{\iota}\bar{\boldsymbol Q}^{(\iota)}(f)-\kappa^{\iota}\bar{\boldsymbol Q}^{(\iota)}(f_g)\|_2. \nonumber
\end{align}
Then, consider the third term in the right side of the above inequality. We had ${\bar{\boldsymbol Q}^{(\iota)}(f)\in\mathbb{C}^{1\times L}}$ and for any ${\boldsymbol v\in\mathbb{C}^{1\times L}}$,${\|\boldsymbol v\|\leq\sqrt{L}\|\boldsymbol v\|_{\infty}}$. The ${j}$th entry of ${\bar{\boldsymbol Q}^{(\iota)}(f)}$ is
\begin{align}
|\kappa^{\iota}\bar{\boldsymbol Q}_j^{(\iota)}(f)| \leq |<\bar{\boldsymbol D}^{-1}\bar{\boldsymbol G}_{\iota}(f),\boldsymbol \Phi_{:,j}>|  \leq8\sqrt{K}\left(256\sqrt{K}\right)=CK\leq CN^2.\nonumber
\end{align}
Next, take ${\kappa^{\iota}\bar{\boldsymbol Q}_j^{(\iota)}(f)}$ as a polynomial of ${z=e^{-i2\pi f}}$ with degree ${m}$ and apply the Bernstein polynomial inequality as
\begin{align}
 &|\kappa^{\iota}\bar{\boldsymbol Q}_j^{(\iota)}(f)-\kappa^{\iota}\bar{\boldsymbol Q}_j^{(\iota)}(f_g)|\leq |e^{-i2\pi f}-e^{-i2\pi f_g}|\underset{z}{sup}\left|\frac{d\kappa^{\iota}\bar{\boldsymbol Q}_j^{(\iota)}(z)}{dz}\right| \nonumber\\
& \leq |e^{-i\pi(f+f_g)}2sin(\pi(-f+f_g))|m\underset{f}{sup}|\kappa^{\iota}\bar{\boldsymbol Q}_j^{(\iota)}(f)| \leq CN^3|f-f_g|.\nonumber
\end{align}
Thus,
\begin{align}
& \|\kappa^{\iota}\bar{\boldsymbol Q}^{(\iota)}(f)-\kappa^{\iota}\bar{\boldsymbol Q}^{(\iota)}(f_g)\|_2\leq\sqrt{L}\|\kappa^{\iota}\bar{\boldsymbol Q}^{(\iota)}(f)-\kappa^{\iota}\bar{\boldsymbol Q}^{(\iota)}(f_g)\|_{\infty} \leq C\sqrt{L}N^3|f-f_g|.\nonumber
\end{align}
The above calculations reveal that the grid size $|\mathbb{T}_{grid}|=1/|f-f_g|$ should be such that ${|f-f_g|\leq\frac{10^{-2}}{4C\sqrt{L}N^3}}$. Using the same arguments, one can obtain the same bound for ${ \|\kappa^{\iota}{\boldsymbol Q}^{(\iota)}(f)-\kappa^{\iota}{\boldsymbol Q}^{(\iota)}(f_g)\|_2}$. Combining the above results with (\ref{eqn40}) proves the lemma.
\subsection{Proof of Lemma\ref{lemma3}}
Consider ${A_{far}=[0,1]\backslash A_{near}}$, where ${\ A_{near}}$ is defined in Lemma \ref{lemma3}. We prove that ${\|\bar{\boldsymbol Q}(f)\|_2<0.99}$ in ${A_{far}}$. Next, it is shown that ${\|\boldsymbol Q(f)\|_2<1}$ in ${A_{near}}$. For ${\|\bar{\boldsymbol Q}(f)\|_2}$, we write
\begin{align}
\|\bar{\boldsymbol Q}(f)\|_2&\leq \sum_{f_k\in\mathbb{T}}^{ }\|\boldsymbol \alpha_k\|_2|\bar{K}(f-f_k)|+\sum_{f_k\in\mathbb{T}}^{}\kappa\|\boldsymbol \beta_k\|_2|\bar{K}^{\prime}(f-f_k)| \nonumber\\
& \leq \|\boldsymbol \alpha\|_{\infty,2}\sum_{f_k\in\mathbb{T}}^{ }|\bar{K}(f-f_k)|+ \|\boldsymbol \beta\|_{\infty,2}\sum_{f_k\in\mathbb{T}}^{}\kappa|\bar{K}^{\prime}(f-f_k)|.\nonumber
\end{align}
Using Lemma H.10 from \cite{fernandez2017demixing}, we have
 $\sum_{j=1}^{K}\kappa^{\iota}|\bar{K}^{(\iota)}(f-f_j)|\leq127C_1+2.42C_2\nonumber$
for some properly chosen ${C_1}$ and ${C_2}$. Thus,
\begin{align}
& \|\bar{\boldsymbol Q}(f)\|_2\leq(\|\boldsymbol \alpha\|_{\infty,2}+\|\boldsymbol \beta\|_{\infty,2})(127C_1+2.42C_2).\nonumber
\end{align}
In the following, we calculate the upper bounds for ${\|\boldsymbol \alpha\|_{\infty,2}}$ and ${\|\boldsymbol \beta\|_{\infty,2}}$. Recall (\ref{eqn26}) for the deterministic case. Using this equation, we have
\begin{align}
& \left[
\begin{array}{c}
\boldsymbol \alpha \\
\boldsymbol \beta \\
\end{array}
\right]=\left[
\begin{array}{c}
\boldsymbol I \\
\bar{\boldsymbol D}_2^{-1}\bar{\boldsymbol D}_1 \\
\end{array}
\right]\bar{\boldsymbol D}_3^{-1}\boldsymbol \Phi,\nonumber
\end{align}
where ${\bar{\boldsymbol D}_3\triangleq\bar{\boldsymbol D}_0+\bar{\boldsymbol D}_1\bar{\boldsymbol D}_2^{-1}\bar{\boldsymbol D}_1}$. According to Lemma 4.1 from \cite{fernandez2016super} and the fact that ${\|\boldsymbol \Phi\|_{\infty,2}=1}$, we have$ \|\boldsymbol \alpha\|_{\infty,2}=\|\bar{\boldsymbol D}_3^{-1}\boldsymbol \Phi\|_{\infty,2}\leq1+2.37\times 10^{-2}$
and
 $\|\boldsymbol \beta\|_{\infty,2}\leq\|\bar{\boldsymbol D}_2^{-1}\bar{\boldsymbol D}_1\bar{\boldsymbol D}_3^{-1}\boldsymbol \Phi\|_{\infty,2} \leq \frac{4.247}{m}\times 10^{-2} .$
Therefore, by the proper choices of ${C_1}$ and ${C_2}$ we get $${\|\bar{\boldsymbol Q}(f)\|_2< 0.99\:\:\: for\:f\in A_{far}}.$$ In order to show that ${\|\boldsymbol Q(f)<1\|_2}$ in ${A_{near}}$, it is enough to show that the second derivative of ${\|\boldsymbol Q(f)<1\|_2}$ is negative in ${A_{near}}$. In a mathematical fashion, it is enough to prove the following inequality,
\begin{align}
& \frac{1}{2}\frac{d^2\|\boldsymbol Q(f)\|_2}{df^2} =\|\boldsymbol Q^{\prime}\|_2^2+Re\{\boldsymbol Q^{\prime \prime}\boldsymbol Q^H(d)\}<0.\label{eqn41}
\end{align}
Now, we investigate each term in the above inequality. For the first term, we can write
\begin{align}
& \|\kappa\boldsymbol Q^{\prime}(f)\|_2^2=\|\kappa\boldsymbol Q^{\prime}(f)-\kappa\bar{\boldsymbol Q}^{\prime}(f)+\kappa\bar{\boldsymbol Q}^{\prime}(f)\|_2^2 \nonumber\\
& \leq 10^{-4}+2\times 10^{-2}\|\kappa \bar{\boldsymbol Q}^{\prime}(f)\|_2+\|\kappa\bar{\boldsymbol Q}^{\prime}(f)\|_2^2,\nonumber
\end{align}
which by applying the kernel bounds of \cite{fernandez2016super} leads to
\begin{align}
 \|\kappa \bar{\boldsymbol Q}^{\prime}(f)\|_2 &\leq\|\boldsymbol \alpha\|_{\infty,2}\sum_{k=1}^{K}\kappa|\bar{K}^{\prime}(f-f_k)|+\|\boldsymbol \beta\|_{\infty,2}\sum_{k=1}^{K}\kappa^2|\bar{K}^{\prime \prime}(f-f_k)| \nonumber\\
&\leq 1.0237\times 2.409\times 10^{-2}+\frac{4.247\times10^{-2}}{m}(0.087)\leq0.0247,\nonumber
\end{align}
where the last inequality is achieved using ${m\geq10^3}$. The second term of (\ref{eqn41}) can be represented as
\begin{align}
&Re\left\{\kappa^2\boldsymbol Q^{\prime \prime}(f)\boldsymbol Q^H(f)\right\}=Re\left\{\kappa^2(\boldsymbol Q^{\prime \prime}(f)-\bar{\boldsymbol Q}^{\prime \prime}(f))\boldsymbol Q^H(f)\right\} \nonumber\\
& +Re\left\{\kappa^2\bar{\boldsymbol Q}^{\prime \prime}(f)(\boldsymbol Q(f)-\bar{\boldsymbol Q}(f))^H\right\}+Re\left\{\kappa^2\bar{\boldsymbol Q}^{\prime \prime}(f)\bar{\boldsymbol Q}^H(f)\right\} \nonumber\\
&  \leq 0.0101+ 0.01+Re\left\{\kappa^2\bar{\boldsymbol Q}^{\prime \prime}(f)\bar{\boldsymbol Q}^H(f)\right\}.\nonumber
\end{align}
Next, we inspect the term ${\kappa^2\bar{\boldsymbol Q}^{\prime \prime}(f)\bar{\boldsymbol Q}^H(f)}$. According to (\ref{eqn28}), we get  $${\kappa^2\bar{\boldsymbol Q}^{\prime \prime}(f)\bar{\boldsymbol Q}^H(f)=\kappa^2\bar{Q}^{\prime \prime}(f)\boldsymbol b^H\boldsymbol b\bar{ Q}^{\ast}(f)=\kappa^2\bar{Q}^{\prime \prime}(f)\bar{ Q}^{\ast}(f)},$$
which is a scalar value. Also, note that
\begin{align}
& \kappa^2Re \left\{\bar{Q}^{\prime \prime}(f)\bar{ Q}^{\ast}(f)\right\}=\kappa^2\left(\bar{Q}_R^{\prime\prime}(f)\bar{Q}_R(f)+|\bar{Q}_I^{\prime\prime}(f)||\bar{Q}_I|\right) \nonumber\\
& \leq \left(-0.8915\times 2.015+0.0474\times 2.555 \right)\leq-1.6752.\nonumber
\end{align}
Thus,
\begin{align}
& \frac{\kappa^2}{2}\frac{d^2\|\boldsymbol Q(f)\|_2}{df^2}\leq-1.6752+0.0201+12.01\times 10^{-4}<0\nonumber
\end{align}
and the proof is complete.

\bibliography{main_R4blackr2}

\end{document}